\documentclass{vldb}

\usepackage{enumitem}
\usepackage{framed}
\usepackage{cprotect}
\usepackage{enumitem}
\usepackage{listings}
\usepackage{amstext}
\usepackage{amstext}
\usepackage{pdfpages}
\usepackage{alltt}
\usepackage{epstopdf}
\usepackage{xspace,colortbl}
\usepackage[USenglish]{babel}
\usepackage{multirow}
\usepackage{url}
\usepackage{subfigure}
\usepackage{graphicx}
\usepackage{amssymb}
\usepackage{fmtcount}
\usepackage{amsfonts}
\usepackage{xspace}
\usepackage{amsmath}
\usepackage{multirow}
\usepackage[mathscr]{eucal}
\usepackage{balance}
\usepackage{colortbl}
\usepackage{bm}
\usepackage{times}
\usepackage[nospace]{cite}
\usepackage{csquotes}

\begin{document}


\newtheorem{theorem}{Theorem}
\newtheorem{example}{Example}
\newtheorem{definition}{Definition}
\newtheorem{problem}{Problem}
\newtheorem{property}{Property}
\newtheorem{proposition}{Proposition}
\newtheorem{lemma}{Lemma}
\newtheorem{corollary}{Corollary}

\newcommand{\cond}{\textrm{pred}\xspace}
\newcommand{\dataset}{data set\xspace}
\newcommand{\datasets}{data sets\xspace}
\newcommand{\spview}{\textsf{SPView}\xspace}
\newcommand{\fjview}{\textsf{FJView}\xspace}
\newcommand{\aggview}{\textsf{AggView}\xspace}
\newcommand{\hashfunc}[1]{\textsf{hash}(#1)\xspace}
\newcommand{\hashop}{\textsf{hash}\xspace}
\newcommand{\nsc}{\textsf{NormalizedSC}\xspace}
\newcommand{\rsc}{\textsf{RawSC}\xspace}

\newcommand{\avgfunc}{\ensuremath{\texttt{avg} }\xspace}
\newcommand{\maxfunc}{\ensuremath{\texttt{max} }\xspace}
\newcommand{\minfunc}{\ensuremath{\texttt{min} }\xspace}
\newcommand{\histfunc}{\ensuremath{\texttt{histogram\_numeric} }\xspace}
\newcommand{\countfunc}{\ensuremath{\texttt{count}}\xspace}
\newcommand{\sumfunc}{\ensuremath{\texttt{sum} }\xspace}
\newcommand{\varfunc}{\ensuremath{\texttt{var} }\xspace}
\newcommand{\stdfunc}{\ensuremath{\texttt{std} }\xspace}
\newcommand{\covfunc}{\ensuremath{\texttt{cov} }\xspace}
\newcommand{\corrfunc}{\ensuremath{\texttt{corr} }\xspace}
\newcommand{\medfunc}{\ensuremath{\texttt{median} }\xspace}
\newcommand{\percfunc}{\ensuremath{\texttt{percentile} }\xspace}
\newcommand{\havingfunc}{\ensuremath{\texttt{HAVING} }\xspace}
\newcommand{\selectfunc}{\ensuremath{\texttt{select} }\xspace}
\newcommand{\ratio}{\ensuremath{\rho }\xspace}

\newcommand{\insertion}{\ensuremath{\texttt{INSERT} }\xspace}
\newcommand{\update}{\ensuremath{\texttt{UPDATE} }\xspace}
\newcommand{\delete}{\ensuremath{\texttt{DELETE} }\xspace}

\newcommand{\svcfull}{Stale View Cleaning\xspace}
\newcommand{\svc}{SVC\xspace}
\newcommand{\svcnospace}{SVC}

\newcommand{\tbl}[1]{\textsf{#1}\xspace}
\newcommand{\field}[1]{\textsf{#1}\xspace}
\newcommand{\cost}{\textrm{cost}\xspace}
\newcommand{\ans}{\textsf{ans}\xspace}
\newcommand{\dans}{\Delta\textsf{ans}\xspace}
\newcommand{\cqp}{correction query processing\xspace}
\newcommand{\Cqp}{Correction query processing\xspace}

\newcommand{\reminder}[1]{{{\textcolor{magenta}{\{\{\bf #1\}\}}}\xspace}}
\newcommand{\specialcell}[2][c]{%
  \begin{tabular}[#1]{@{}c@{}}#2\end{tabular}}

\def\ojoin{\setbox0=\hbox{$\bowtie$}%
  \rule[-.02ex]{.25em}{.4pt}\llap{\rule[\ht0]{.25em}{.4pt}}}
\def\leftouterjoin{\mathbin{\ojoin\mkern-5.8mu\bowtie}}
\def\rightouterjoin{\mathbin{\bowtie\mkern-5.8mu\ojoin}}
\def\fullouterjoin{\mathbin{\ojoin\mkern-5.8mu\bowtie\mkern-5.8mu\ojoin}}

\pagestyle{plain}

\title{Stale View Cleaning: Getting Fresh Answers from Stale Materialized Views}

\numberofauthors{1}
\author{\large Sanjay Krishnan, Jiannan Wang, Michael J. Franklin, Ken Goldberg, Tim Kraska{$\,^\dag$} \\
\vspace{.2em}\affaddr{\large UC Berkeley, ~~ $^\dag$Brown University} \\
\vspace{.1em}\affaddr{\large \{sanjaykrishnan, jnwang, franklin, goldberg\}@berkeley.edu}\\
\affaddr{\large tim\_kraska@brown.edu}
}

\fontsize{9pt}{10.5pt}
\selectfont

\maketitle

\vspace{-1em}

\begin{abstract}
Materialized views (MVs), stored pre-computed results, are widely used to facilitate fast queries on large datasets. When new records arrive at a high rate, it is infeasible to continuously update (maintain) MVs and a common solution is to defer maintenance by batching updates together. Between batches the MVs become increasingly stale with incorrect, missing, and superfluous rows leading to increasingly inaccurate query results.
We propose \svcfull (\svc) which addresses this problem from a data cleaning perspective. 
In \svc, we efficiently clean a sample of rows from a stale MV, and use the clean sample to estimate aggregate query results.
While approximate, the estimated query results reflect the most recent data.
As sampling can be sensitive to long-tailed distributions, we further explore an outlier indexing technique to give increased accuracy when the data distributions are skewed. \svc complements existing deferred maintenance approaches by giving accurate and bounded query answers between maintenance. We evaluate our method on a generated dataset from the TPC-D benchmark and a real video distribution application.  Experiments confirm our theoretical results: (1) cleaning an MV sample is more efficient than full view maintenance, (2) the estimated results are more accurate than using the stale MV, and (3) \svc is applicable for a wide variety of MVs.
\end{abstract}

\pagestyle{empty} 

\vspace{-0.5em}
\section{Introduction}
Storing pre-computed query results, also known as materialization, is an extensively studied approach to reduce query latency on large data \cite{LarsonY85, gupta1995maintenance, chirkova2011materialized}. 
Materialized Views (MVs) are now supported by all major commercial vendors.
However, as with any pre-computation or caching, the key challenge in using MVs is maintaining their freshness as base data changes.
While there has been substantial work in incremental maintenance of MVs \cite{chirkova2011materialized, DBLP:journals/vldb/KochAKNNLS14}, eager maintenance (i.e., immediately applying updates) is not always feasible.

In applications such as monitoring or visualization \cite{rainbird, DBLP:journals/cgf/LiuJH13}, analysts may create many MVs by slicing or aggregating over different dimensions.
Eager maintenance requires updating all affected MVs for every incoming transaction, and thus, each additional MV reduces the available transaction throughput.
This problem becomes significantly harder when the views are distributed and computational resources are contended by other tasks.
As a result, in production environments, it is common to batch updates together to amortize overheads \cite{chirkova2011materialized}.
Batch sizes are set according to system constraints, and can vary from a few seconds to even nightly.  

While increasing the batching period gives the user more flexibility to schedule around system constraints, a disadvantage is that MVs are stale between maintenance periods.
Other than an educated guess based on past data, the user has no way of knowing how incorrect their query results are.
Some types of views and query workloads can be sensitive to even a small number of base data updates, for example, if updates disproportionately affect a subset of frequently queried rows.
Thus, any amount of staleness is potentially dangerous, and this presents us a dichotomy between facing the cost of eager maintenance or coping with consequences of unknown inaccuracy.
In this paper, we explore an intriguing middle ground, namely, we can derive a bounded approximation of the correct answer for a fraction of the cost. 
With a small amount of up-to-date data, we can compensate for the error in aggregate query results induced by staleness.


\begin{figure}[t]
\centering
 \includegraphics[scale=0.29]{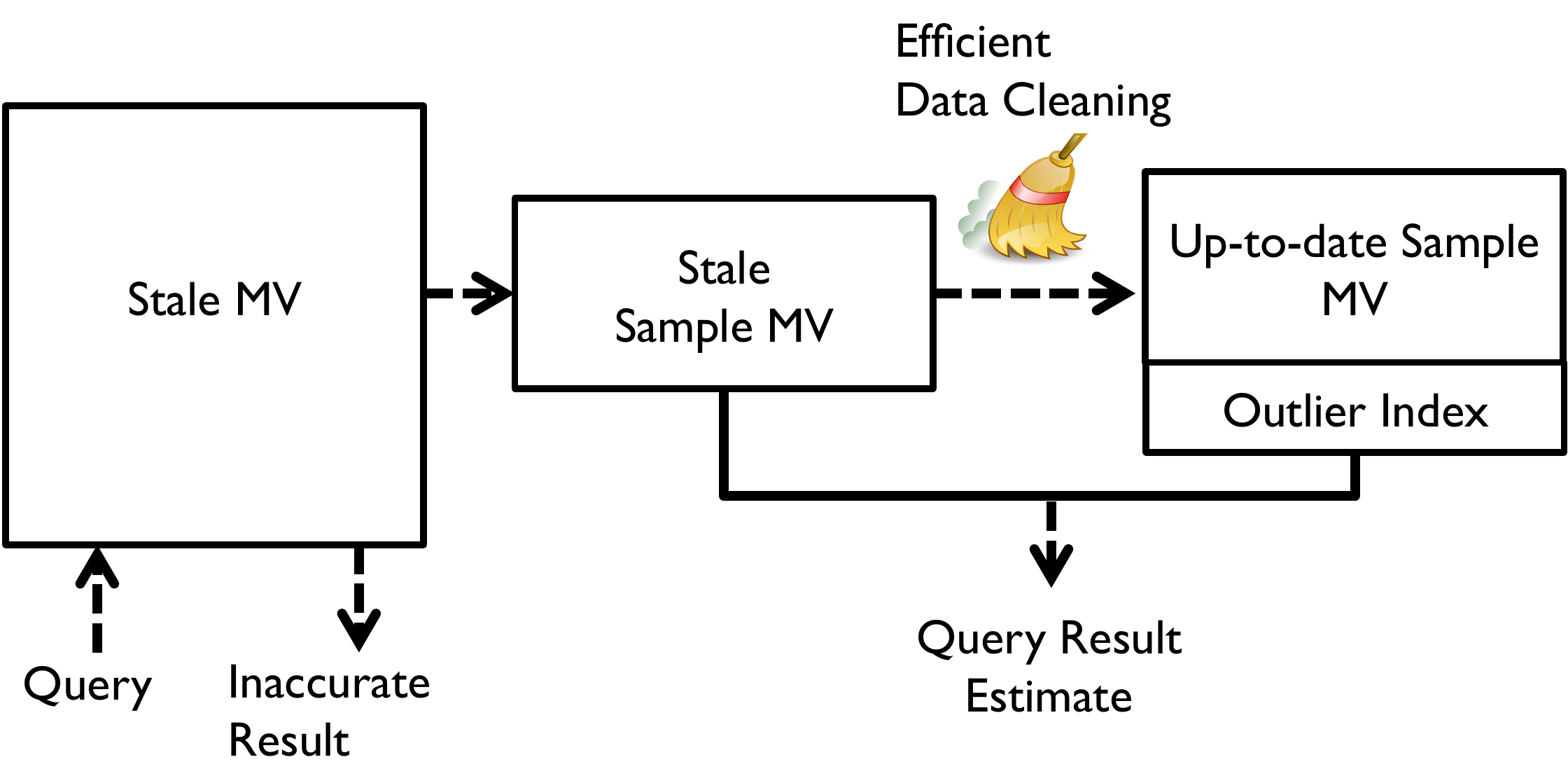} \vspace{-.25em}
 \caption{In \svc, we pose view maintenance as a sample-and-clean problem and show that we can use a sample of clean (up-to-date) rows from an MV to correct inaccurate query results on stale views.\label{sys-arch}}\vspace{-1.75em}
\end{figure}

Our method relies on modeling query answering on stale MVs as a data cleaning problem.
A stale MV has incorrect, missing, or superfluous rows, which are problems that have been studied in the data cleaning literature (e.g., see Rahm and Do for a survey\cite{rahm2000data}).
Increasing data volumes have led to development of new, efficient sampling-based approaches for coping with dirty data.   
In our prior work, we developed the SampleClean framework for scalable aggregate query processing on dirty data \cite{wang1999sample}.
Since data cleaning is often expensive, we proposed cleaning a sample of data and using this sample to improve the results of aggregate queries on the full dataset.
Since stale MVs are dirty data, an approach similar to SampleClean raises a new possibility of using a sample of ``clean'' rows in the MVs to return more accurate query results.

\vspace{0.5em}

\svcfull (\svc illustrated in Figure~\ref{sys-arch}) approximates aggregate query results from a stale MV and a small sample of up-to-date data.
We calculate a relational expression that materializes a uniform sample of up-to-date rows.
This expression can be interpreted as ``cleaning" a stale sample of rows.
We use the clean sample of rows to estimate a result for an aggregate query on the view.
The estimates from this procedure, while approximate, reflect the most recent data. 
Approximation error is more manageable than staleness because: (1) the uniformity of sampling allows us to apply theory from statistics such as the Central Limit Theorem to give tight bounds on approximate results, and (2) the approximate error is parametrized by the sample size which the user can control trading off accuracy for computation.
However, the MV setting presents new challenges that we did not consider in prior work.
To summarize our contributions:(1) a hashing-based technique that efficiently materializes an up-to-date sample view, (2) algorithms for processing general aggregate queries on a sample view and bounding results in confidence intervals, (3) an outlier indexing technique to reduce sensitivity to skewed datasets that can push the index up to derived relations, and (4) an evaluation of this technique on real and synthetic datasets to show that SVC gives highly accurate results for a relatively small maintenance cost.

The paper is organized as follows: 
In Section~\ref{sec-background}, we give the necessary background for our work.
Next, in Section~\ref{sec-arch}, we formalize the problem.
In Sections~\ref{sampling} and~\ref{correction}, we describe the sampling and query processing of our technique.
In Section~\ref{outlier}, we describe the outlier indexing framework.
Then, in Section~\ref{exp}, we evaluate our approach.
We discuss Related Work in Section~\ref{related}.
Finally, we present our Conclusions in Section~\ref{conclusion}.

\section{Background}\label{sec-background}

\subsection{Motivation and Example}\label{subsec-inc}
Materialized view maintenance can be very expensive resulting in staleness. 
Many important use-cases require creating a large number of views including: visualization, personalization, privacy, and real-time monitoring.  
The problem with eager maintenance is that every view created by an analyst places a bottleneck on incoming transactions.
There has been significant research on fast MV maintenance algorithms, most recently DBToaster \cite{DBLP:journals/vldb/KochAKNNLS14} which uses SQL query compilation and higher-order maintenance.
However, even with these optimizations, some materialized views are computationally difficult to incrementally maintain.
For example, incremental maintenance of views with correlated subqueries can grow with the size of the data.
It is also common to use the same infrastructure to maintain multiple MVs (along with other analytics tasks) adding further contention to computational resources and reducing overall available throughput. 
When faced with such challenges, one solution is to batch updates and amortize maintenance overheads.

\noindent \textbf{Log Analysis Example: } 
Suppose we are a video streaming company analyzing user engagement.
Our database consists of two tables \tbl{Log} and \tbl{Video}, with the following schema:
\begin{lstlisting}[mathescape,basicstyle={\scriptsize}]
Log(sessionId$\textrm{,}$ videoId)
Video(videoId$\textrm{,}$ ownerId$\textrm{,}$ duration)
\end{lstlisting}
The \tbl{Log} table stores each visit to a specific video with primary key (\texttt{sessionId}) and a foreign-key to the \tbl{Video} table (\texttt{videoId}).
For our analysis, we are interested in finding aggregate statistics on visits, such as the average visits per video and the total number of visits predicated on different subsets of owners. 
We could define the following MV that counts the visits for each \texttt{videoId} associated with owners and the duration.

\vspace{2em}

\begin{lstlisting}[mathescape,basicstyle={\scriptsize}]
CREATE VIEW visitView
AS SELECT videoId, ownerId, duration, 
count(1) as visitCount
FROM Log, Video WHERE Log.videoId = Video.videoId
GROUP BY videoId
\end{lstlisting}
As \tbl{Log} table grows, this MV becomes stale, and we denote the insertions to the table as:
\begin{lstlisting}[mathescape,basicstyle={\scriptsize}]
LogIns(sessionId$\textrm{,}$ videoId)
\end{lstlisting}

Staleness does not affect every query uniformly.
Even when the number of new entries in \texttt{LogIns} is small relative to \texttt{Log}, some queries might be very inaccurate.
For example, views to newly added videos may account for most of \texttt{LogIns}, so queries that count visits to the most recent videos will be more inaccurate.
The amount of inaccuracy is unknown to the user, who can only estimate an expected error based on prior experience.
This assumption may not hold in rapidly evolving data.
We see an opportunity for approximation through sampling which can give bounded query results for a reduced maintenance cost.
In other words, a small amount of up-to-date data allows the user to estimate the magnitude of query result error due to staleness.

\subsection{SampleClean~\cite{wang1999sample}}
SampleClean is a framework for scalable aggregate query processing on dirty data.
Traditionally, data cleaning has explored expensive, up-front cleaning of entire datasets for increased query accuracy.
Those who were unwilling to pay the full cleaning cost avoided data cleaning altogether.
We proposed SampleClean to add an additional trade-off to this design space by using sampling, i.e., bounded results for aggregate queries when only a sample of data is cleaned.
The problem of high computational costs for accurate results mirrors the challenge faced in the MV setting with the tradeoff between immediate maintenance (expensive and up-to-date) and deferred maintenance (inexpensive and stale). 
Thus, we explore how samples of ``clean" (up-to-date) data can be used for improved query processing on MVs without incurring the full cost of maintenance.

However, the metaphor of stale MVs as a Sample-and-Clean problem only goes so far and there are significant new challenges that we address in this paper.
In prior work, we modeled data cleaning as a row-by-row black-box transformation.
This model does not work for missing and superfluous rows in stale MVs.
In particular, our sampling method has to account for this issue and we propose a hashing based technique to efficiently materialize a uniform sample even in the presence of missing/superfluous rows.
Next, we greatly expand the query processing scope of SampleClean beyond \sumfunc, \countfunc, and \avgfunc queries.
Bounding estimates that are not \sumfunc, \countfunc, and \avgfunc queries, is significantly more complicated.
This requires new analytic tools such as a statistical bootstrap estimation to calculate confidence intervals.
Finally, we add an outlier indexing technique to improve estimates on skewed data.

\section{Framework Overview}\label{sec-arch}

\subsection{Notation and Definitions}\label{notation}
\svc returns a bounded approximation for aggregate queries on stale MVs for a flexible additional maintenance cost.

\noindent \textbf{Materialized View:} Let $\mathcal{D}$ be a database which is a collection of relations $\{R_i\}$. A \emph{materialized view} $S$ is the result of applying a \emph{view definition} to $\mathcal{D}$. 
View definitions are composed of standard relational algebra expressions: Select ($\sigma_{\phi}$), Project ($\Pi$), Join ($\bowtie$), Aggregation ($\gamma$), Union ($\cup$), Intersection ($\cap$) and Difference ($-$). 
We use the following parametrized notation for joins, aggregations and generalized projections:
\begin{itemize}[noitemsep] \sloppy
	\item $\Pi_{a_1,a_2,...,a_k}(R)$: Generalized projection selects attributes $\{a_1,a_2,...,a_k\}$ from $R$, allowing for adding new attributes that are arithmetic transformations of old ones (e.g., $a_1+a_2$).
	\item $\bowtie_{\phi (r1,r2)}(R_1,R_2)$: Join selects all tuples in $R_1 \times R_2$ that satisfy $\phi (r_1,r_2)$. We use $\bowtie$ to denote all types of joins even extended outer joins such as $\rightouterjoin,\leftouterjoin,\fullouterjoin$.
	\item $\gamma_{f,A}(R)$: Apply the aggregate function $f$ to the relation R grouped by the distinct values of $A$, where $A$ is a subset of the attributes.  
	The DISTINCT operation can be considered as a special case of the Aggregation operation. 
\end{itemize}
\vspace{-0.5em}
The composition of the unary and binary relational expressions can be represented as a tree, which is called the \emph{expression tree}.
The leaves of the tree are the \emph{base relations} for the view. 
Each non-leave node is the result of applying one of the above relational expressions to a relation.
To avoid ambiguity, we refer to tuples of the base relations as \emph{records} and tuples of derived relations as \emph{rows}.

\vspace{0.5em}

\noindent \textbf{Primary Key:} We assume that each of the base relations has a \emph{primary key}. If this is not the case, we can always add an extra column 
that assigns an increasing sequence of integers to each record.
For the defined relational expressions, every row in a materialized view can also be given a primary key \cite{DBLP:journals/vldb/CuiW03, DBLP:conf/sigmod/ZengGMZ14},
which we will describe in Section \ref{sampling}. 
This primary key is formally a subset of attributes $u \subseteq \{a_1,a_2,...,a_k\}$ such that all $s \in S(u)$ are unique.

\vspace{0.5em}

\noindent \textbf{Staleness:} For each relation $R_i$ there is a set of insertions $\Delta R_i$ (modeled as a relation)
and a set of deletions $\nabla R_i$.
An ``update'' to $R_i$ can be modeled as a deletion and then an insertion.
We refer to the set of insertion and deletion relations as ``delta relations", denoted by $\partial \mathcal{D}$:
\[
	\partial \mathcal{D} = \{\Delta R_1,...,\Delta R_k\} \cup \{\nabla R_1,...,\nabla R_k\}
\]
A view $S$ is considered \emph{stale} when there exist insertions or deletions to any of its base relations.
This means that at least one of the delta relations in $\partial \mathcal{D}$ is non-empty.

\vspace{0.5em}

\noindent \textbf{Maintenance:} There may be multiple ways (e.g., incremental maintenance or recomputation) to maintain a view $S$, and we denote the up-to-date view as $S'$.
We formalize the procedure to maintain the view as a \emph{maintenance strategy} $\mathcal{M}$.
A maintenance strategy is a relational expression the execution of which will return $S'$.
It is a function of the database $\mathcal{D}$, the stale view $S$, and all the insertion and deletion relations $\partial \mathcal{D}$. 
In this work, we consider maintenance strategies composed of the same relational expressions as materialized views described above.
\[
S' = \mathcal{M}(S,\mathcal{D}, \partial D)
\]

\vspace{0.5em}

\noindent \textbf{Staleness as Data Error:} The consequences of staleness are incorrect, missing, and superfluous rows. 
Formally, for a stale view $S$ with primary key $u$ and an up-to-date view $S'$:
\begin{itemize}[noitemsep] \sloppy
	\item \textbf{Incorrect: } Incorrect rows are the set of rows (identified by the primary key) that are updated in $S'$. For $s \in S$, let $s(u)$ be the value of the primary key. An incorrect row is one such that there exists a $s' \in S'$ with $s'(u) = s(u)$ and $s \ne s'$.
	\item \textbf{Missing: } Missing rows are the set of rows (identified by the primary key) that exist in the up-to-date view but not in the stale view. For $s' \in S'$, let $s'(u)$ be the value of the primary key. A missing row is one such that there does not exist a $s \in S$ with $s(u) = s'(u)$.
	\item \textbf{Superfluous: } Superfluous rows are the set of rows (identified by the primary key) that exist in the stale view but not in the up-to-date view. For $s \in S$, let $s(u)$ be the value of the primary key. A superfluous row is one such that there does not exist a $s' \in S'$ with $s(u) = s'(u)$.
\end{itemize}

\vspace{2.0em}

\noindent \textbf{Uniform Random Sampling:}
We define a sampling ratio $m\in [0,1]$ and for each row in a view $S$, we include it into a sample with probability $m$.
We use the ``hat'' notation (e.g., $\widehat{S}$) to denote sampled relations.
The relation $\widehat{S}$ is a \emph{uniform sample} of $S$ if
\[\text{(1) } \forall s \in \widehat{S} : s \in S\text{;~~~~~ (2) }Pr(s_1 \in \widehat{S}) =  Pr(s_2 \in \widehat{S}) = m.\]

\vspace{0.5em}

We say a sample is \emph{clean} if and only if it is a uniform random sample of the up-to-date view $S'$. 

\vspace{-0.5em}

\begin{example}\label{concepts}
In this example, we summarize all of the key concepts and terminology pertaining to materialized views, stale data error, and maintenance strategies.
Our example view, visitView, joins the Log table with the Video table and counts the visits for each video grouped by videoId.
Since there is a foreign key relationship between the relations, this is just a visit count for each unique video with additional attributes. 
The primary keys of the base relations are: sessionId for Log and videoId for Video.

If new records have been added to the Log table, the visitView is considered stale.
Incorrect rows in the view are videos for which the visitCount is incorrect and missing rows are videos that had not yet been viewed once at the time of materialization. 
While not possible in our running example, superfluous rows would be videos whose Log records have all been deleted.
Formally, in this example our database is $\mathcal{D}=\{Video, Log\}$, and the delta relations are $\partial\mathcal{D}=\{LogIns\}$. 

Suppose, we apply the change-table IVM algorithm proposed in~\cite{gupta1995maintenance}:
\vspace{-.55em}
\begin{enumerate}[noitemsep]
\item Create a ``delta view" by applying the view definition to LogIns. That is, calculate the visit count per video on the new logs:
\[
 \gamma(Video \bowtie LogIns)
\]
\item Take the full outer join of the ``delta view" with the stale view visitView (equality on videoId).
\[
 VisitView \fullouterjoin \gamma(Video \bowtie LogIns)
\]
\item Apply the generalized projection operator to add the visitCount in the delta view to each of the rows in visitView where we treat a NULL value as 0: 
\[
 \Pi (VisitView \fullouterjoin \gamma(Video \bowtie LogIns))
\]
Therefore, the maintenance strategy is:
\[
 \mathcal{M}(\{VisitView\},\{Video, Log\}, \{LogIns\})
\]
\[
\text{\hspace{0.7em}} = \Pi (VisitView \fullouterjoin \gamma(Video \bowtie LogIns))
\]
\end{enumerate}

\end{example}



\subsection{\svc Workflow}
Formally, the workflow of \svc is:\vspace{-0.5em}
\begin{enumerate}[noitemsep]
\item We are given a view $S$.
\item $\mathcal{M}$ defines the maintenance strategy that updates $S$ at each maintenance period.
\item The view $S$ is stale between periodic maintenance, and the up-to-date view should be $S'$.
\item \emph{(Problem 1. Stale Sample View Cleaning)} We find an expression $\mathcal{C}$ derived from $\mathcal{M}$ 
that cleans a uniform random sample of the stale view $\widehat{S}$ to produce a ``clean" sample of the up-to-date
view $\widehat{S'}$.
\item \emph{(Problem 2. Query Result Estimation)} Given an aggregate query $q$ and the state query result $q(S)$, we use $\widehat{S'}$ and $\widehat{S}$ to estimate the up-to-date result.
\item We optionally maintain an index of outliers $o$ for improved estimation in skewed data.
\end{enumerate} 

\noindent\textbf{Stale Sample View Cleaning: }
The first problem addressed in this paper is how to clean a sample of the stale materialized view.
\begin{problem}[Stale Sample View Cleaning]\sloppy
We are given a stale view $S$, a sample of this stale view $\widehat{S}$ with ratio $m$, the maintenance strategy $\mathcal{M}$, the base relations $\mathcal{D}$, and
the insertion and deletion relations $\partial \mathcal{D}$.
We want to find a relational expression $\mathcal{C}$ such that:
\[
\widehat{S}' = \mathcal{C}(\widehat{S},\mathcal{D},\partial \mathcal{D}),
\]
where $\widehat{S}'$ is a sample of the up-to-date view with ratio $m$. 
\end{problem}

\noindent\textbf{Query Result Estimation: }
The second problem addressed in this paper is query result estimation.
\begin{problem}[Query Result Estimation]
Let $q$ be an aggregate query of the following form \footnote{\scriptsize For simplicity, we exclude the group by clause for all queries in the paper, as it can be modeled as part of the \textsf{Condition}.}:
\begin{lstlisting} [mathescape,basicstyle={\scriptsize}]
SELECT $agg(a)$ FROM View WHERE Condition(A);
\end{lstlisting}
If the view $S$ is stale, then the result will be incorrect by some value~$c$:
\[
q(S') = q(S) + c
\]
Our objective is to find an estimator $f$ such that:
\[
q(S') \approx f(q(S),\widehat{S},\widehat{S}')
\] 
\end{problem}

\begin{example}\label{infexample}
Suppose a user wants to know how many videos have received more than 100 views.
\begin{lstlisting}[basicstyle={\scriptsize}]
SELECT COUNT(1) FROM visitView WHERE visitCount > 100;
\end{lstlisting}
Let us suppose the user runs the query and the result is $45$.
However, there have now been new records inserted into the Log table making this result stale. 
First, we take a sample of \tbl{visitView} and suppose this sample is a 5\% sample.
In Stale Sample View Cleaning (Problem 1), we apply updates, insertions, and deletions to the sample to efficiently materialize a 5\%  sample of the up-to-date view.
In Query Result Estimation (Problem 2), we estimate aggregate query results based on the stale sample and the up-to-date sample.
\end{example}

\section{Efficiently Cleaning a Sample} \label{sampling}
In this section, we describe how to find a relational expression $\mathcal{C}$ derived from the maintenance strategy $\mathcal{M}$ that
efficiently cleans a sample of a stale view $\widehat{S}$ to produce $\widehat{S}'$.

\subsection{Challenges}
To setup the problem, we first consider two naive solutions to this problem that will not work. 
We could trivially apply $\mathcal{M}$ to the entire stale view $S$ and update it to $S'$, and then sample.
While the result is correct according to our problem formulation, it does not save us on any computation for maintenance.
We want to avoid materialization of up-to-date rows outside of the sample. 
However, the naive alternative solution is also flawed. 
For example, we could just apply $\mathcal{M}$ to the stale sample $\widehat{S}$ and a sample of the delta relations $\widehat{\partial \mathcal{D}}$. 
The challenge is that $\mathcal{M}$ does not always commute with sampling. 

\subsection{Provenance}
\label{lin}
To understand the commutativity problem, consider maintaining a group by aggregate view:
\begin{lstlisting} [mathescape,basicstyle={\scriptsize}]
SELECT videoId, count(1) FROM Log
GROUP BY videoId
\end{lstlisting}
The resulting view has one row for every distinct \texttt{videoId}.
We want to materialize a sample of $S'$, that is a sample of distinct \texttt{videoId}.
If we sample the base relation \texttt{Log} first, we do not get a sample of the view.
Instead, we get a view where every count is partial.

To achieve a sample of $S'$, we need to ensure that for each $s \in S'$ all contributing rows in subexpressions to $s$ are also sampled. 
This is a problem of row provenance \cite{DBLP:journals/vldb/CuiW03}.
Provenance, also termed lineage, has been an important tool in the analysis of materialized views \cite{DBLP:journals/vldb/CuiW03} and in approximate query processing \cite{DBLP:conf/sigmod/ZengGMZ14}.
\begin{definition}[Provenance]\label{prov}
Let $r$ be a row in relation $R$, let $R$ be derived from some other
relation $R = exp(U)$ where $exp(\cdot)$ be a relational
expression composed of the expressions defined in Section \ref{notation}.
The provenance of row $r$ with respect to $U$ is $p_U(r)$. 
This is defined as the set of rows in $U$ such that for an update to any row $u \not\in p_U(r)$, it guarantees that $r$ is unchanged.
\end{definition}

\subsection{Primary Keys}
For the relational expressions defined in the previous sections, this provenance is well defined and can be tracked using primary key rules that are enforced on
each subexpression \cite{DBLP:journals/vldb/CuiW03}. 
We recursively define a set of primary keys for all relations in the expression tree:
\begin{definition} [Primary Key Generation]\label{pk}
For every relational expression $R$, we define the primary key attribute(s) of every expression to be:
\begin{itemize}[noitemsep]
\item Base Case: All relations (leaves) must have an attribute $p$ which is designated as a primary key. 
\item $\sigma_{\phi}(R)$: Primary key of the result is the primary key of R 
\item $\Pi_{(a_1,...,a_k)}(R)$: Primary key of the result is the primary key of R. The primary key must always be included in the projection.
\item $\bowtie_{\phi (r1,r2)}(R_1,R_2)$: Primary key of the result is the tuple of the primary keys of $R_1$ and $R_2$. 
\item $\gamma_{f,A}(R)$: The primary key of the result is the group by key $A$ (which may be a set of attributes).
\item $R_1 \cup R_2$: Primary key of the result is the union of the primary keys of $R_1$ and $R_2$
\item $R_1 \cap R_2$: Primary key of the result is the intersection of the primary keys of $R_1$ and $R_2$
\item $R_1 - R_2$: Primary key of the result is the primary key of $R_1$
\end{itemize}
For every node at the expression tree, these keys are guaranteed to uniquely identify a row.
\end{definition}
These rules define a constructive definition that can always be applied for our defined relational expressions. 


\begin{figure}[t]
\centering
 \includegraphics[scale=0.20]{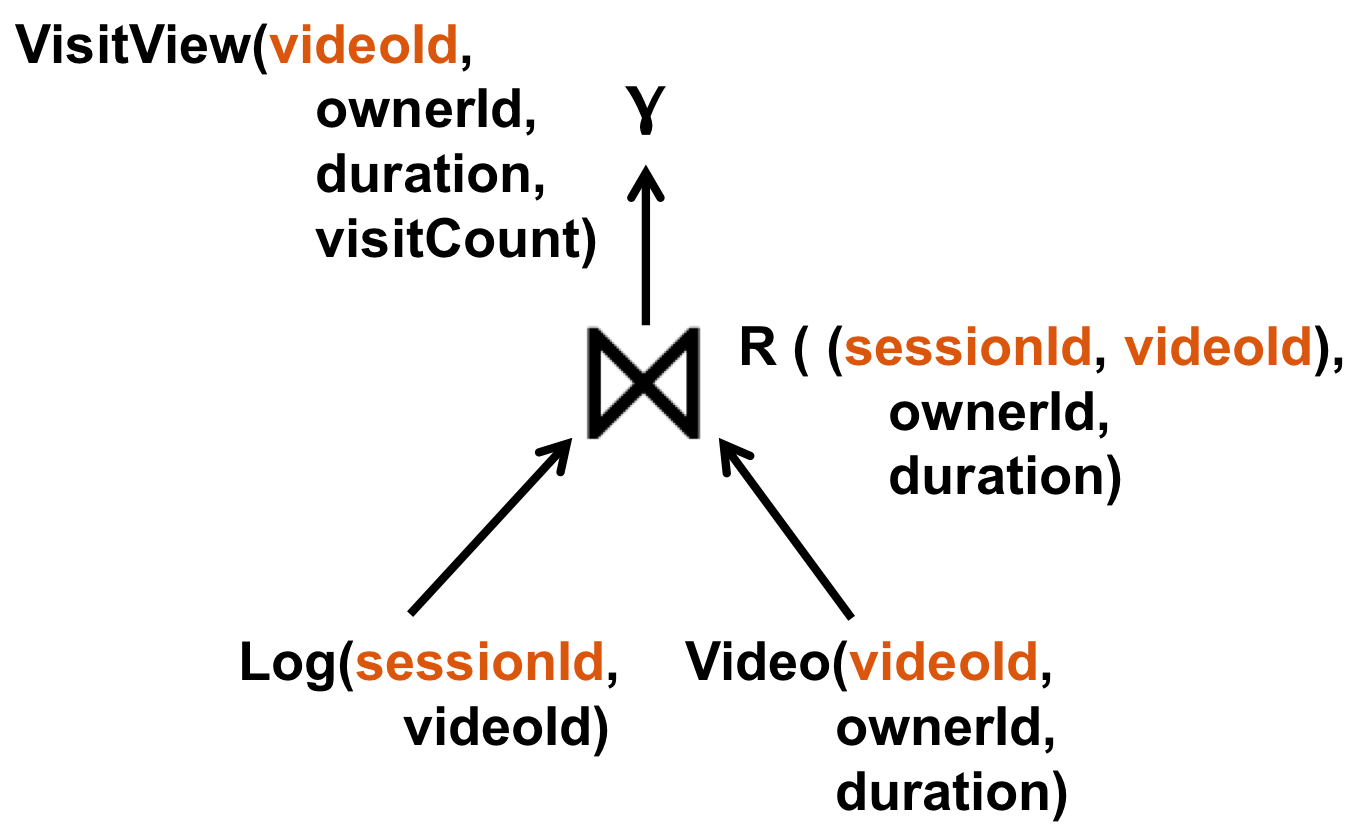} \vspace{-.5em}
 \caption{Applying the rules described in Definition \ref{pk}, we illustrate how to assign a primary key to a view.  \label{pk-fig}}\vspace{-1.5em}
\end{figure}

\begin{example}
A variant of our running example view that does not have a primary key is:
\begin{lstlisting}[mathescape,basicstyle={\scriptsize}]
CREATE VIEW visitView AS SELECT count(1) as visitCount
FROM Log, Video WHERE Log.videoId = Video.videoId
GROUP BY videoId
\end{lstlisting}
We illustrate the key generation process in Figure \ref{pk-fig}.
Suppose there is a base relation, such as \tbl{Log}, that is missing a primary key (sessionId)\footnote{\scriptsize It does not make sense for Video to be missing a primary key in our running example due to the foreign key relationship}.
We can add this attribute by generating an increasing sequence of integers for each record in \tbl{Log}. 
Since both base tables \tbl{Video} and \tbl{Log} have primary keys videoId and sessionId respectively,
the result of the join will have a primary key (videoId, sessionId).
Since the group by attribute is videoId, that becomes the primary key of the view.
\end{example}

\subsection{Hashing Operator}
\label{push}
The primary keys allow us to determine the set of rows that contribute to a row $r$ in a derived relation.
If we have a deterministic way of mapping a primary key to a Boolean, we can ensure that all contributing rows are also sampled. 
To achieve this we use a hashing procedure.
Let us denote the hashing operator $\eta_{a, m}(R)$. 
For all tuples in R, this operator applies a hash function whose range is $[0,1]$ to primary key $a$ (which may be a set) and selects those records with hash less than or equal to $m$ \footnote{\scriptsize For example, if hash function is a 32-bit unsigned integer which we can normalize by \texttt{MAXINT} to be in $[0,1]$.}.

In this work, we study uniform hashing where the condition $h(a) \le m$ implies that a fraction of approximately $m$ of the rows are sampled.
Such hash functions are utilized in other aspects of database research and practice (e.g. hash partitioning, hash joins, and hash tables).
Hash functions in these applications are designed to be as uniform as possible to avoid collisions.
Numerous empirical studies establish that many commonly applied hash functions (e.g., Linear, SDBM, MD5, SHA) have negligible differences with a true uniform random variable \cite{henke2009empirical, l2007testu01}.
Cryptographic hashes work particularly well and are supported by most commercial and open source systems, for example MySQL provides MD5 and SHA1.


To avoid materializing extra rows, we push down the hashing operator through the expression tree.
The further that we can push $\eta$ down, the more operators (i.e., above the sampling) can benefit.
This push-down is analogous to predicate push-down operations used in query optimizers. 
In particular, we are interested in finding an optimized relational expression that materializes an identical sample before and after the push-down.
We formalize the push-down rules below:
\begin{definition}[Hash push-down]
For a derived relation $R$, the following rules can be applied to push $\eta_{a, m}(R)$ down the expression tree. 
\begin{itemize}[noitemsep]
\item $\sigma_{\phi}(R)$: Push $\eta$ through the expression.  
\item $\Pi_{(a_1,...,a_k)}(R)$: Push $\eta $ through if $a$ is in the projection.
\item $\bowtie_{\phi (r1,r2)}(R_1,R_2)$: No push down in general. There are special cases below where push down is possible.
\item $\gamma_{f,A}(R)$: Push $\eta $ through if $a$ is in the group by clause $A$.
\item $R_1 \cup R_2$: Push $\eta $ through to both $R_1$ and $R_2$
\item $R_1 \cap R_2$: Push $\eta $ through to both $R_1$ and $R_2$
\item $R_1 - R_2$: Push $\eta $ through to both $R_1$ and $R_2$
\end{itemize}
\end{definition}

\noindent \textbf{Special Case of Joins: }
In general, a join $R \bowtie S$ blocks the push-down of the hash operator $\eta_{a, m}(R)$ since $a$ possibly consists of attributes in both $R$ and $S$.
However, when there is a constraint that enforces these attributes are equal then push-down is possible.

\emph{Foreign Key Join. } If we have a join with two foreign-key relations $R_1$ (fact table with foreign key $a$) and $R_2$ (dimension table with primary key $b \subseteq a$) and we are sampling the key $a$, then we can push the sampling down to $R_1$. This is because we are guaranteed that for every $r_1\in R_1$ there is only one $r_2 \in R_2$. 

\emph{Equality Join. } If the join is an equality join and $a$ is one of the attributes in the equality join condition $R_1.a = R_2.b$, then $\eta$ can be pushed down to both $R_1$ and $R_2$. On $R_1$ the pushed down operator is $\eta_{a, m}(R_1)$ and on $R_2$ the operator is $\eta_{b, m}(R_2)$. 

\begin{example}
We illustrate our hashing procedure in terms of SQL expressions on our running example.
We can push down the hash function for the following expressions:
\begin{lstlisting}[mathescape]
SELECT * FROM Video WHERE Condition($\cdot$)
SELECT * FROM Video,Log WHERE Video.videoId = Log.videoId
SELECT videoId, count(1) FROM Log GROUP BY videoId
\end{lstlisting}
\vspace{0.3em}
The following expressions are examples where we cannot push-down the hash function:
\begin{lstlisting}
SELECT * FROM Video, Log

SELECT c, count(1)
FROM ( 
       SELECT videoId, count(1) as c FROM Log 
       GROUP BY videoId
     )
GROUP BY c
\end{lstlisting}

\end{example}

In Theorem \ref{push-down-corr}, we prove the correctness of our push-down rules.

\begin{theorem}\label{push-down-corr}
Given a derived relation $R$, primary key $a$, and the sample $\eta_{a, m}(R)$.
Let $S$ be the sample created by applying $\eta_{a, m}$ without push-down and 
$S'$ be the sample created by applying the push-down rules to $\eta_{a, m}(R)$.
$S$ and $S'$ are identical samples with sampling ratio $m$.
\end{theorem}
\begin{proof}[Sketch]
We can prove this by induction.
The base case is where the expression tree is only one node, trivially making this true.
Then, we can induct considering one level of operators in the tree.
$\sigma, \cup, \cap, -$ clearly commutes with hashing $a$.
$\Pi$ commutes only if $a$ is in the projection.
For $\bowtie$, a sampling operator on $Q$ can be pushed down if $a$ is in either $k_r$ or $k_s$, or if there is a constraint that links $k_r$ to $k_s$.
For group by aggregates, if $a$ is in the group clause (i.e., it is in the aggregate), then hashing the operand filters all rows that have $a$ which is sufficient to materialize the derived row.
\end{proof}

\subsection{Efficient View Cleaning}
If we apply the hashing operator to $\mathcal{M}$, we can get an optimized cleaning expression $\mathcal{C}$ that avoids materializing unnecessary rows. 
When applied to a stale sample of a view $\widehat{S}$, the database $\mathcal{D}$, and the delta relations $\partial \mathcal{D}$, it produces an up-to-date sample with sampling ratio $m$:
\[
\widehat{S}' = \mathcal{C}(\widehat{S},\mathcal{D},\partial \mathcal{D})
\]
Thus, it addresses Problem 1 from the previous section.

\begin{figure}[t]
\centering
 \includegraphics[scale=0.22]{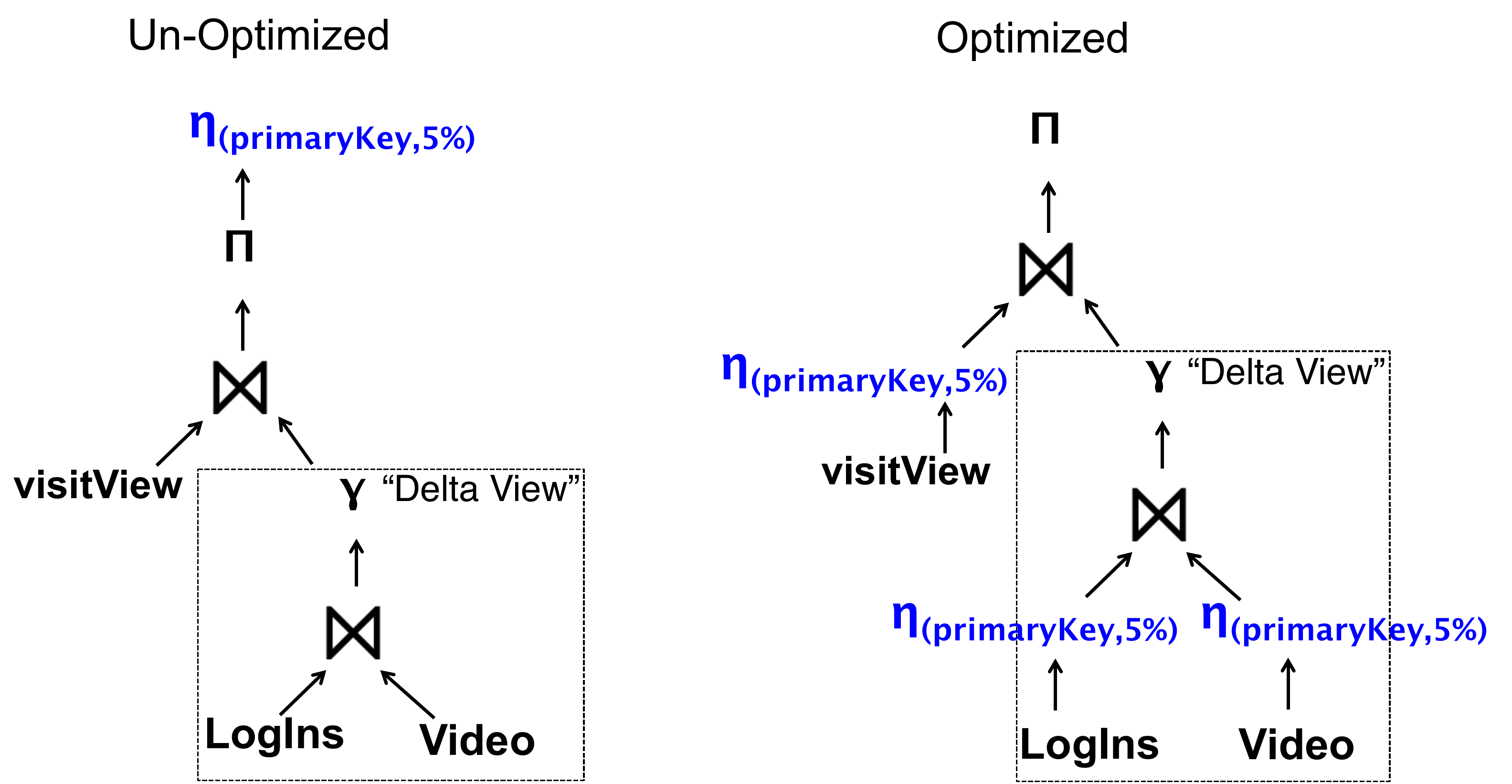} \vspace{-.5em}
 \caption{Applying the rules described in Section \ref{push}, we illustrate how to optimize the sampling of our example maintenance strategy. \label{exexpr2}}\vspace{-1em}
\end{figure}

\begin{example}
We illustrate our proposed approach on our example view \texttt{visitView} with the expression tree listed in Figure \ref{exexpr2}. 
We start by applying the hashing operator to the primary key (\texttt{videoId}).
The next operator we see in the expression tree is a projection that increments the \texttt{visitCount} in the view, and this allows
for push-down since primary key is in the projection.
The second expression is a hash of the equality join key which merges the aggregate from the ``delta view'' to the old view allowing us to push down on both branches of the tree using our special case for equality joins.
On the left side, we reach the stale view so we stop.
On the right side, we reach the aggregate query (count) and since the primary key is in group by clause, we can push down the sampling.
Then, we reach another point where we hash the equality join key allowing us to push down the sampling to the relations \tbl{LogIns} and \tbl{Video}.
\end{example}

\subsection{Corresponding Samples}
We started with a uniform random sample $\widehat{S}$ of the stale view $S$.
The hash push down allows us to efficiently materialize the sample $\widehat{S}'$.
$\widehat{S}'$ is a uniform random sample of the up-to-date view S.
While both of these samples are uniform random samples of their respective relations, 
the two samples are correlated since $\widehat{S}'$ is generated by cleaning $\widehat{S}$.
In particular, our hashing technique ensures that the primary keys in $\widehat{S}'$ depend on the primary keys in $\widehat{S}$.
Statistically, this positively correlates the query result $q(\widehat{S}')$ and $q(\widehat{S})$. 
We will see how this property can be leveraged to improve query estimation accuracy (Section \ref{re}). 
\vspace{-0.5em}
\begin{property}[Correspondence]
Suppose $\widehat{S'}$ and $\widehat{S}$ are uniform samples of $S'$ and $S$, respectively.  Let $u$ denote the primary key. We say $\widehat{S'}$ and $\widehat{S}$ correspond if and only if:
\vspace{0.5em}
\begin{itemize}[noitemsep]
\item Uniformity: $\widehat{S'}$ and $\widehat{S}$ are uniform random samples of $S'$ and $S$ respectively with a sampling ratio of $m$
\item Removal of Superfluous Rows: $D = \{\forall s \in \widehat{S} \; \nexists s' \in S': s(u) = s'(u)\}$, $D \cap \widehat{S'} = \emptyset$ 
\item Sampling of Missing Rows: $I = \{\forall s' \in \widehat{S'} \; \nexists s \in S: s(u) = s'(u)\}$, $\mathbb{E}(\mid I \cap \widehat{S'} \mid) = m\mid I \mid $ 
\item Key Preservation for Updated Rows: For all $s\in \widehat{S}$ and not in $D$ or $I$, $s' \in \widehat{S}': s'(u) = s(u)$.
\end{itemize}
\vspace{-.25em}
\label{correspondence}
\end{property}
\vspace{-1.0em}

\section{Query Result Estimation}
\label{correction}
\svc returns two corresponding samples, $\widehat{S}$ and $\widehat{S}'$.
$\widehat{S}$ is a ``dirty" sample (sample of the stale view) and $\widehat{S}'$ is a ``clean" sample (sample of the up-to-date view).
In this section, we first discuss how to estimate query results using the two corresponding samples. 
Then, we discuss the bounds and guarantees on different classes of aggregate queries.


\subsection{Result Estimation}\label{re}
Suppose, we have an aggregate query $q$ of the following form:
\begin{lstlisting} [mathescape]
$\;\;\;\;$q(View) := SELECT f(attr) FROM View WHERE cond(*)
\end{lstlisting}
We quantify the staleness $c$ of the aggregate query result as the difference
between the query applied to the stale view $S$ compared to the up-to-date view $S'$:
\[
q(S') = q(S) + c
\]
The objective of this work is to estimate $q(S')$.
In the Approximate Query Processing (AQP) literature, sample-based estimates have been well studied \cite{OlkenR86, AgarwalMPMMS13}.
This inspires our first estimation algorithm, \svcnospace+AQP, which uses \svc to materialize a sample view and an AQP-style
result estimation technique.

\vspace{0.25em}

\noindent\textbf{SVC+AQP: }  Given a clean sample view $\widehat{S}'$, the query $q$, and a scaling factor $s$, 
we apply the query to the sample and scale it by $s$:
\[
q(S') \approx s \cdot q(\widehat{S}')
\]
For example, for the \sumfunc and \countfunc the scaling factor is $\frac{1}{m}$. For the \avgfunc the scaling factor is 1.
Refer to \cite{OlkenR86, AgarwalMPMMS13} for a detailed discussion on the scaling factors.

\svcnospace+AQP returns what we call a direct estimate of $q(S')$.
We could, however, try to estimate $c$ instead.
Since we have the stale view $S$, we could run the query $q$ on the full stale view and 
estimate the difference $c$ using the samples $\widehat{S}$ and $\widehat{S}'$.
We call this approach \svcnospace+CORR, which represents calculating a correction to $q(S)$ instead of a direct estimate.

\vspace{0.25em}

\noindent\textbf{SVC+CORR: } Given a clean sample $\widehat{S}'$, its corresponding dirty sample $\widehat{S}$, a query q, and a scaling factor $s$:
\begin{enumerate}[noitemsep]
\item Apply \svcnospace+AQP to $\widehat{S}'$:
$r_{est\_fresh} = s \cdot q(\widehat{S}') $
\item Apply \svcnospace+AQP to $\widehat{S}$:
$r_{est\_stale} = s \cdot q(\widehat{S}) $ 
\item Apply q to the full stale view:
$r_{stale} = q(S) $
\item Take the difference between (1) and (2) and add it to (3):
\[
q(S') \approx r_{stale} + (r_{est\_fresh} - r_{est\_stale})
\]
\end{enumerate}

A commonly studied property in the AQP literature is unbiasedness.
An unbiased result estimate means that the expected value of the estimate over all samples of the same size is $q(S')$ \footnote{\scriptsize The \avgfunc query is considered conditionally unbiased in some works.}.
We can prove that if \svcnospace+AQP is unbiased (there is an AQP method that gives an unbiased result) then \svcnospace+CORR also gives unbiased results.
\begin{lemma}\label{lemma:unbiased}
If there exists an unbiased sample estimator for q(S') then there exists an unbiased sample estimator for c.
\end{lemma}
\begin{proof}[Sketch] 
Suppose, we have an unbiased sample estimator $e_q$ of $q$. 
Then, it follows that $\mathbb{E}\big[e_q(\widehat{S'})\big] = q(S')$
If we substitute in this expression:
$c = \mathbb{E}\big[e_q(\widehat{S'})\big] - q(S) $.
Applying the linearity of expectation:
$ c = \mathbb{E}\big[e_q(\widehat{S'}) - q(S)\big] $
\end{proof}
Some queries do not have unbiased sample estimators, but the bias of their sample estimators can be bounded. Example queries include: \medfunc, \percfunc.
A corollary to the previous lemma, is that if we can bound the bias for our estimator then we can achieve a bounded bias for $c$ as well.

\begin{example}
We can formalize our earlier example query in Section \ref{infexample} in terms of SVC+CORR and SVC+AQP.
Let us suppose the initial query result is $45$.
There now have been new log records inserted into the Log table making the old result stale, and suppose we are working with a sampling ratio of 5\%.
For SVC+AQP, we count the number of videos in the clean sample that currently have counts greater than 100 and scale that result by $\frac{1}{5\%} = 20$. 
If the count from the clean sample is $4$, then the estimate for SVC+AQP is $80$.
For SVC+CORR, we also run SVC+AQP on the dirty sample.
Suppose that there are only two videos in the dirty sample with counts above 100, then the result of running SVC+AQP on the dirty sample is $20\cdot2 = 40$.
We take the difference of the two values $80 - 40 = 40$.
This means that we should correct the old result by $40$ resulting in the estimate of $45+40 = 85$.
\end{example}

\subsection{Confidence Intervals}
To bound our estimates in confidence intervals we explore three cases: (1) aggregates that can be written as sample means, (2) aggregates that can be bounded empirically with a statistical bootstrap, and (3) \minfunc and \maxfunc.
For (1), \sumfunc, \countfunc, and \avgfunc can all be written as sample means.
\sumfunc is the sample mean scaled by the relation size and \countfunc is the mean of the indicator function scaled by the relation size.
In this case, we can get analytic confidence intervals which allows us to analyze the efficiency tradeoffs.
In case (2), for example \medfunc, we lose this property and have to use an empirical technique to bound the results.
Queries such as \minfunc and \maxfunc fall into their own category as they cannot easily be bounded empirically \cite{agarwalknowing}, and we discuss these queries in our Technical Report \cite{technicalReport}.


\subsubsection{Confidence Intervals For Sample Means}
The first case is aggregates that can be expressed as a sample mean (\sumfunc, \countfunc, and \avgfunc)
Sample means for uniform random samples (also called sampling without replacement) converge to the population mean by the Central Limit Theorem (CLT).
Let $\bar{\mu}$ be a sample mean calculated from $k$ samples, $\sigma^2$ be the variance of the sample, and $\mu$ be the population mean. 
Then, the error $(\mu - \bar{\mu})$ is normally distributed:
$
 N(0,\frac{\sigma^2}{k})
$.
Therefore, the confidence interval is given by:
\[
\bar{\mu} \pm \gamma \sqrt{\frac{\sigma^2}{k}}
\]
where $\gamma$ is the Gaussian tail probability value (e.g., 1.96 for 95\%, 2.57 for 99\%).

We discuss how to calculate this confidence interval in SQL for SVC+AQP.
The first step is a query rewriting step where we move the predicate \textsf{cond(*)} into the SELECT clause (1 if true, 0 if false). 
Let \emph{attr} be the aggregate attribute and $m$ be the sampling ratio. 
We define an intermediate result $trans$ which is a table of transformed rows with the first column the 
primary key and the second column defined in terms of \texttt{cond(*)} statement and scaling.
For \sumfunc:
\begin{lstlisting}[mathescape,basicstyle={\scriptsize}]
trans= SELECT pk,1.0/m$\cdot$attr$\cdot$cond(*) as trans_attr FROM s 
\end{lstlisting} 
For \countfunc:
\begin{lstlisting}[mathescape,basicstyle={\scriptsize}]
trans= SELECT pk, 1.0/m $\cdot$ cond(*) as trans_attr FROM s
\end{lstlisting}
For \avgfunc since there is no scaling we do not need to re-write the query:
\begin{lstlisting}[mathescape,basicstyle={\scriptsize}]
trans= SELECT pk,attr as trans_attr FROM s WHERE cond(*) 
\end{lstlisting}

\vspace{0.25em}

\noindent\textbf{SVC+AQP: } The confidence interval on this result is defined 
as:
\begin{lstlisting}[mathescape,basicstyle={\scriptsize}]
SELECT $\gamma\cdot$stdev(trans_attr)/sqrt(count(1)) FROM trans
\end{lstlisting}

\vspace{0.25em}

To calculate the confidence intervals for \svcnospace+CORR we have to look at the statistics of the difference, i.e., $c = q(S) - q(S')$, from a sample.
If all rows in $\widehat{S}$ exist in $\widehat{S}'$, we could use the associativity of addition and subtraction to rewrite this as:
$c = q(S - S')$, where $-$ is the row-by-row difference between $S$ and $S'$.
The challenge is that the missing rows on either side make this ill-defined.
Thus, we defined the following null-handling with a subtraction operator we call $\dot{-}$.
\begin{definition}[Correspondence Subtract] Given an aggregate query, and two corresponding relations $R_1$ and $R_2$ with the schema $(a_1, a_2, ...)$ where $a_1$ is the primary key for $R_1$ and $R_2$, and $a_2$ is the aggregation attribute for the query. 
$\dot{-}$ is defined as a projection of the full outer join on equality of $R_1.a_1 = R_2.a_1$: \[ \Pi_{R_1.a_2 - R_2.a_2} ( R_1 \fullouterjoin R_2 ) \]
Null values $\emptyset$ are represented as~zero.
\end{definition}
Using this operator, we can define a new intermediate result $diff$:
\[diff := trans(\widehat{S}') \dot{-} trans(\widehat{S}) \]

\vspace{0.35em}\noindent\textbf{SVC+CORR: } Then, as in \svcnospace+AQP, we bound the result using the CLT:
\begin{lstlisting}[mathescape,basicstyle={\scriptsize}]
SELECT $\gamma\cdot$stdev(trans_attr)/sqrt(count(1)) FROM diff
\end{lstlisting}

\subsubsection{AQP vs. CORR For Sample Means}\label{aqp-v-cor}
In terms of these bounds, we can analyze how \svcnospace+AQP compares to \svcnospace+CORR for a fixed sample size $k$.
\sloppy
\svcnospace+AQP gives an estimate that is proportional to the variance of the clean sample view: 
$\frac{\sigma_{S'}^2}{k}$.
\svcnospace+CORR to the variance of the \emph{differences}: 
$\frac{\sigma_{c}^2}{k}$.
Since the change is the difference between the stale and up-to-date view, this can be rewritten as
\[\frac{\sigma_{S}^2 + \sigma_{S'}^2 - 2cov(S,S')}{k}\]
Therefore, a correction will have less variance when:
\[\sigma_{S}^2 \le 2cov(S,S')\]

As we saw in the previous section, correspondence correlates the samples.
If the difference is small, i.e., $S$ is nearly identical to $S'$, then $cov(S,S')~\approx~\sigma_{S}^2$. 
This result also shows that there is a point when updates to the stale MV are significant enough that direct estimates are more accurate.
When we cross the break-even point we can switch from using \svcnospace+CORR to \svcnospace+AQP.
\svcnospace+AQP does not depend on $cov(S,S')$ which is a measure of how much the data has changed.
Thus, we guarantee an approximation error of at most $\frac{\sigma_{S'}^2}{k}$.
In our experiments (Figure \ref{exp-1-total}(b)), we evaluate this break even point empirically. 

\subsubsection{Selectivity For Sample Means}
Let $p$ be the selectivity of the query and $k$ be the sample size; that is, a fraction $p$ records from the relation satisfy the predicate.
For these queries, we can model selectivity as a reduction of effective sample size $k\cdot p$ making the
estimate variance: $O(\frac{1}{k*p})$.
Thus, the confidence interval's size is scaled up by $\frac{1}{\sqrt{p}}$.
Just like there is a tradeoff between accuracy and maintenance cost, for a fixed accuracy, 
there is also a tradeoff between answering more selective queries and maintenance cost.

\subsubsection{Optimality For Sample Means}
Optimality in unbiased estimation theory is defined in terms of the variance of the estimate \cite{cox1979theoretical}.
\begin{proposition}
An estimator is called a minimum variance unbiased estimator (MVUE) if it is unbiased and the variance of the estimate is less than or equal to that of any other unbiased estimate.
\end{proposition}

A sampled relation $R$ defines a discrete distribution. It is important to note that this distribution is different from the data generating distribution, since even if $R$ has continuous valued attributes $R$ still defines a discrete distribution. Our population is finite and we take a finite sample thus every sample takes on only a discrete set of values. In the general case, this distribution is only described by the set of all of its values (i.e., no smaller parametrized representation). In this setting, the sample mean is an MVUE. In other words, if we make no assumptions about the underlying distribution of values in $R$, SVC+AQP and SVC+CORR are optimal for their respective estimates ($q(S')$ and $c$). Since they estimate different variables, even with optimality SVC+CORR might be more accurate than SVC+AQP and vice versa. 
There are, however, some cases when the assumptions, namely zero-knowledge, of this optimality condition do not hold.
 As a simple counter example, if we knew our data were exactly on a line, a sample size of two is sufficient to answer any aggregate query. 
 However, even for many parametric distributions, the sample mean estimators are still MVUEs, e.g., poisson, bernouilli, binomial, normal, and exponential. It is often difficult and unknown in many cases to derive an MVUE other than a sample mean. 
 Our approach is valid for any choice of estimator if one exists, even though we do the analysis for sample mean estimators and this is the setting in which that estimator is optimal.

\vspace{1.5em}

\subsubsection{Bootstrap Confidence Intervals}
In the second case, we explore bounding queries that cannot be expressed as sample means.
We do not get analytic confidence intervals on our results, nor is it guaranteed that our estimates are optimal.
In AQP, the commonly used technique is called a statistical bootstrap \cite{AgarwalMPMMS13} to empirically bound the results.
In this approach, we repeatedly subsample with replacement from our sample and apply the query to the sample.
This gives us a technique to bound SVC+AQP the details of which can be found in \cite{AgarwalMPMMS13, agarwalknowing, DBLP:conf/sigmod/ZengGMZ14}. 
For SVC+CORR, we have to propose a variant of bootstrap to bound the estimate of $c$.
In this variant, repeatedly estimate $c$ from subsamples and build an empirical distribution for $c$.

\vspace{0.35em}

\noindent\textbf{SVC+CORR: } To use bootstrap to find a 95\% confidence interval:
\begin{enumerate}[noitemsep]
\item Subsample $\widehat{S}_{sub}'$ and $\widehat{S}_{sub}$ with replacement from $\widehat{S}'$ and $\widehat{S}$ respectively
\item Apply SVC+AQP to $\widehat{S}_{sub}'$ and $\widehat{S}_{sub}$
\item Record the difference $\cdot(aqp(\widehat{S}_{sub}')-aqp(\widehat{S}_{sub}))$
\item Return to 1, for k iterations.
\item Return the $97.5$\% and the $2.5$\% percentile of the distribution of results. 
\end{enumerate}

\vspace{-.5em}
\section{Outlier Indexing}\label{outlier}
Sampling is known to be sensitive to outliers \cite{clauset2009power, chaudhuri2001overcoming}.
Power-laws and other long-tailed distributions are common in practice \cite{clauset2009power}.
The basic idea is that we create an index of outlier records (records whose attributes deviate from the mean value greatly) and ensure that these records are included in the sample, since these records greatly increase the variance of the data. 

\subsection{Indices on the Base Relations}
The first step is that the user selects an attribute of any base relation to index and specifies a threshold $t$ and a size limit $k$.
In a single pass of updates (without maintaining the view), the index is built storing references to the records with attributes greater than $t$.
If the size limit is reached, the incoming record is compared to the smallest indexed record and if it is greater then we evict the smallest record.
The same approach can be extended to attributes that have tails in both directions by making the threshold $t$ a range, which takes the highest and the lowest values.
However, in this section, we present the technique as a threshold for clarity.

There are many approaches to select a threshold.
We can use prior information from the base table, a calculation which can be done in the background during the periodic maintenance cycles.
If our size limit is $k$, for a given attribute we can select the the top-k records with that attributes.
Then, we can use that top-k list to set a threshold for our index. 
Then, the attribute value of the lowest record becomes the threshold $t$.
Alternatively, we can calculate the variance of the attribute and set the threshold to represent $c$ standard deviations above the mean.
This threshold can be adaptively set at each maintenance period.


\subsection{Adding Outliers to the Sample}
Given this index, the next question is how we can use this information in our materialized views.
We need to propagate the indices upwards through the expression tree.
We add the condition that the only eligible indices are ones on base relations that are being sampled (i.e., we can push the hash operator down to that relation).
Therefore, in the same iteration as sampling, we can also test the index threshold and add records to the outlier index. 
We formalize the propagation property recursively. 
Every relation can have an outlier index which is a set of attributes and a set of records that exceed the threshold value on those attributes.
The main idea is to treat the indexed records as a sub-relation that gets propagated upwards with the maintenance strategy.
\begin{definition}[outlier index pushup]
Define an outlier index to be a tuple of a set of indexed attributes, and a set of records $(I,O)$. The outlier index propagates upwards with the following rules: 
\begin{itemize}[noitemsep]
\item Base Relations: Outlier indices on base relations are pushed up only if that relation is being sampled, i.e., if the sampling operator can be pushed down to that relation.
\item $\sigma_{\phi}(R)$: Push up with a new outlier index and apply the selection to the outliers $(I,\sigma_{\phi}(O))$ 
\item $\Pi_{(a_1,...,a_k)}(R)$: Push upwards  $(I \cap (a_1,...,a_k), O)$.
\item $\bowtie_{\phi (r1,r2)}(R_1,R_2)$: Push upwards $(I_{1} \cup I_{2}, O_1 \bowtie O_2)$. 
\item $\gamma_{f,A}(R)$: For group-by aggregates, we set $I$ to be the aggregation attribute. For the outlier index, we do the following steps. (1) Apply the aggregation to the outlier index $\gamma_{f,A}(O)$, (2) for all distinct $A$ in $O$ select the row in $\gamma_{f,A}(R)$ with the same $A$, and (3) this selection is the new set of outliers $O$. 
\item $R_1 \cup R_2$: Push up with a new outlier index $(I_1 \cap I_2, O_1 \cup O_2)$. The set of index attributes is combined with an intersection to avoid missed outliers.
\item $R_1 \cap R_2$: Push up with a new outlier index $(I_1 \cap I_2, O_1 \cap O_2)$.
\item $R_1 - R_2$: Push up with a new outlier index $(I_1 \cup I_2, O_1 - O_2)$.
\end{itemize}
\end{definition}

For all outlier indices that can propagate to the view (i.e., the top of the tree), we get a final set $O$ of records. 
Given these rules, $O$ is, in fact, a subset of our materialized view $S'$.
Thus, our query processing can take advantage of the theory described in the previous section to incorporate the set $O$ into our results.
We implement the outlier index as an additional attribute on our sample with a boolean flag true or false if it is an outlier indexed record.
If a row is contained both in the sample and the outlier index, the outlier index takes precedence.
This ensures that we do not double count the outliers.

\subsection{Query Processing}\label{oqp} 
For result estimation, we can think of our sample $\hat{S'}$ and our outlier index $O$ as two distinct parts.
Since $O \subset S'$, and we give membership in our outlier index precedence, our sample is actually a sample restricted to the set $\widehat{(S'-O)}$.   
For a given query, let $c_{reg}$ be the correction calculated on $\widehat{(S'-O)}$ using the technique proposed in the previous section and adjusting the sampling ratio $m$ to account for outliers removed from the sample.
We can also apply the technique to the outlier set $O$ since this set is deterministic the sampling ratio for this set is $m=1$, and we call this result $c_{out}$.
Let $N$ be the count of records that satisfy the query's condition and $l$ be the number of outliers that satisfy the condition.
Then, we can merge these two corrections in the following way:
$
 v = \frac{N-l}{N}c_{reg} + \frac{l}{N}c_{out}
$.
For the queries in the previous section that are unbiased, this approach preserves unbiasedness.
Since we are averaging two unbiased estimates $c_{reg}$ and $c_{out}$, the linearity of the expectation operator preserves this property.
Furthermore, since $c_{out}$ is deterministic (and in fact its bias/variance is 0), $c_{reg}$ and $c_{out}$ are uncorrelated making the bounds described in the previous section applicable as well.

\begin{example}
We chose an attribute in the base data to index, for example \texttt{duration}, and an example threshold of 1.5 hours.
We apply the rules to push the index up, and this materializes the entire set of rows whose duration is longer than 1.5 hours.
For SVC+AQP, we run the query on the set of clean rows with durations longer than 1.5 hours.
Then, we use the update rule in Section \ref{oqp} to update the result based on the number of records in the index and the total size of the view.
For SVC+CORR, we additionally run the query on the set of dirty rows with durations longer than 1.5 hours and take the difference between SVC+AQP.
As in SVC+AQP, we use the update rule in Section \ref{oqp} to update the result based on the number of records in the index and the total size of the view.
\end{example}


\section{Results}
\label{exp}
We evaluate \svc first on a single node MySQL database to evaluate its accuracy, performance, and efficiency in a variety of materialized view 
scenarios.
Then, we evaluate the outlier indexing approach in terms of improved query accuracy and also evaluate the overhead associated with using the index.
After evaluation on the benchmark, we present an application of server log analysis with a dataset from a video streaming company, Conviva.

\subsection{Experimental Setup} 
\noindent\textbf{Single-node Experimental Setup: }
Our single node experiments are run on a r3.large Amazon EC2 node (2x Intel Xeon E5-2670, 15.25 GB Memory, and 32GB SSD Disk) with a MySQL version 5.6.15 database.
These experiments evaluate views from a 10GB TPCD-Skew dataset.
TPCD-Skew dataset \cite{tpcdskew} is based on the Transaction Processing Council's benchmark
schema (TPCD) but is modified so that it generates a dataset with values drawn from a Zipfian distribution instead of uniformly.
The Zipfian distribution \cite{mitzenmacher2004brief} is a long-tailed distribution with a single parameter $z=\{1,2,3,4\}$ where a larger
value means a more extreme tail and $z=1$ corresponds to the basic TPCD benchmark. 
In our experiments, we use use $z=2$ unless otherwise noted.
The incremental maintenance algorithm used in our experiments is the ``change-table" or ``delta-table" method used in numerous works in incremental maintenance \cite{gupta1995maintenance,gupta2006incremental, DBLP:journals/vldb/KochAKNNLS14}.
In all of the applications, the updates are kept in memory in a temporary table, and we discount this loading time from our experiments.
We build an index on the primary keys of the view, but not on the updates.
Below we describe the view definitions and the queries on the views\footnote{\scriptsize Refer to our extended paper on more details about the experimental setup \cite{technicalReport}.}:

\emph{Join View:} In the TPCD specification, two tables receive insertions and updates: \textsf{lineitem} and \textsf{orders}.
Out of 22 parametrized queries in the specification, 12 are group-by aggregates of the join of \textsf{lineitem} and \textsf{orders} (Q3, Q4, Q5, Q7, Q8, Q9, Q10, Q12, Q14, Q18, Q19, Q21).
Therefore, we define a materialized view of the foreign-key join of \textsf{lineitem} and \textsf{orders}, and compare incremental view maintenance and \svc.
We treat the 12 group-by aggregates as queries on the view.


\emph{Complex Views:} Our goal is to demonstrate the applicability of \svc outside of simple materialized views that include nested queries and other more complex relational algebra.
We take the TPCD schema and denormalize the database, and treat each of the 22 
TPCD queries as views on this denormalized schema. 
The 22 TPCD queries are actually parametrized queries where parameters, such as the selectivity of the predicate, are randomly set by the TPCD \textsf{qgen} program.
Therefore, we use the program to generate 10 random instances of each query and use each random instance as a materialized view.
10 out of the 22 sets of views can benefit from \svc.
For the 12 excluded views, 3 were static (i.e, there are no updates to the view based on the TPCD workload), and the remaining 9 views have a small cardinality not making them suitable for sampling.

For each of the views, we generated \emph{queries on the views}.
Since the outer queries of our views were group by aggregates, we picked a random attribute $a$ from the group by clause and a random attribute $b$ from aggregation.
We use $a$ to generate a predicate.
For each attribute $a$, the domain is specified in the TPCD standard.
We select a random subset of this domain, e.g., if the attribute is country then the predicate can be $\text{countryCode} > 50$ and $\textsf{countryCode} < 100$.
We generated 100 random \sumfunc, \avgfunc, and \countfunc queries for each view.

\noindent\textbf{Distributed Experimental Setup: }
We evaluate \svc on Apache Spark~1.1.0 with 1TB of logs from a video streaming company, Conviva \cite{conviva}.
This is a denormalized user activity log corresponding to video views and various metrics such as data transfer rates, and latencies.
Accompanying this data is a four month trace of queries in SQL.
We identified 8 common summary statistics-type queries that calculated engagement and error-diagnosis metrics.
These 8 queries defined the views in our experiments.
We populated these view definitions using the first 800GB of user activity log records.  
We then applied the remaining 200GB of user activity log records as the updates (i.e., in the order they arrived) in our experiments.
We generated aggregate random queries over this view by taking either random time ranges or random subsets of customers.

\subsubsection{Metrics and Evaluation}

\noindent\textbf{No maintenance (Stale): } The baseline for evaluation is not applying any maintenance to the materialized view.

\noindent\textbf{Incremental View Maintenance (IVM): } We apply incremental view maintenance (change-table based maintenance \cite{gupta1995maintenance,gupta2006incremental, DBLP:journals/vldb/KochAKNNLS14}) to the full view.

\noindent\textbf{\svcnospace+AQP: } We maintain a sample of the materialized view using \svc and estimate the result with AQP-style estimation technique. 

\noindent\textbf{\svcnospace+CORR: } We maintain a sample of the materialized view using \svc and process queries on the view using the correction which applies the AQP to both the clean and dirty samples, and uses both estimates to correct a stale query result.

\vspace{0.5em}
Since \svc has a sampling parameter, we denote a sample size of $x \% $ as \svcnospace+CORR-x or \svcnospace+AQP-x, respectively. 
To evaluate accuracy and performance, we define the following metrics:

\noindent\textbf{Relative Error: } For a query result $r$ and an incorrect result $r'$, the relative error is $\frac{\mid r-r' \mid}{r}.$
When a query has multiple results (a group-by query), then, unless otherwise noted, relative error is defined as the median over all the errors.

\noindent\textbf{Maintenance Time: } We define the maintenance time as the time needed to produce the up-to-date view for incremental view maintenance, and the time needed to produce the up-to-date sample in \svc. 

\subsection{Join View}
In our first experiment, we evaluate how \svc performs on a materialized view of the join of \textsf{lineitem} and \textsf{orders}.
We generate a 10GB base TPCD-Skew dataset with skew $z=2$, and derive the view from this dataset.
We first generate 1GB (10\% of the base data) of updates (insertions and updates to existing records), and vary the sample size.

\textbf{Performance:}
Figure \ref{exp-1-samplesize}(a) shows the maintenance time of \svc as a function of sample size.
With the bolded dashed line, we note the time for full IVM. 
For this materialized view, sampling allows for significant savings in maintenance time; albeit for approximate answers.
While full incremental maintenance takes 56 seconds, \svc with a 10\% sample can complete in 7.5 seconds.

The speedup for \svc-10
In the next figure, Figure \ref{exp-1-samplesize}(b), we evaluate this speedup. 
We fix the sample size to 10\% and plot the speedup of \svc compared to IVM while varying the size of the updates.
On the x-axis is the update size as a percentage of the base data.
For small update sizes, the speedup is smaller, 6.5x for a 2.5\% (250MB) update size.
As the update size gets larger, \svc becomes more efficient, since for a 20\% update size (2GB), the speedup is 10.1x. 
The super-linearity is because this view is a join of \textsf{lineitem} and \textsf{orders} and we assume that there is not a join index on the updates.
Since both tables are growing sampling reduces computation super-linearly. 

\textbf{Accuracy:}
At the same design point with a 10\% sample, we evaluate the accuracy of \svc.
In Figure \ref{exp-1-acc}, we answer TPCD queries with this view.
The TPCD queries are group-by aggregates and we plot the median relative error for \svcnospace+CORR, No Maintenance, and \svcnospace+AQP.
On average over all the queries, we found that \svcnospace+CORR was 11.7x more accurate than the stale baseline, and 3.1x more accurate than applying \svcnospace+AQP to the sample.

\textbf{\svcnospace+CORR vs. \svcnospace+AQP:}
While more accurate, it is true that \svcnospace+CORR moves some of the computation from maintenance to query execution.
\svcnospace+CORR calculates a correction to a query on the full materialized view.
On top of the query time on the full view (as in IVM) there is additional time to calculate a correction from a sample.
On the other hand \svcnospace+AQP runs a query only on the sample of the view.
We evaluate this overhead in Figure \ref{exp-1-total}(a), where we compare the total maintenance time and query execution time.
For a 10\% sample \svcnospace+CORR required 2.69 secs to execute a \sumfunc over the whole view, IVM required 2.45 secs, and  \svcnospace+AQP required 0.25 secs.
However, when we compare this overhead to the savings in maintenance time it is small.

\svcnospace+CORR is most accurate when the materialized view is less stale as predicted by our analysis in Section \ref{aqp-v-cor}.
On the other hand \svcnospace+AQP is more robust to the staleness and gives a consistent relative error.
The error for \svcnospace+CORR grows proportional to the staleness.
In Figure \ref{exp-1-total}(b), we explore which query processing technique, \svcnospace+CORR or \svcnospace+AQP, should be used.
For a 10\% sample, we find that \svcnospace+CORR is more accurate until the update size is 32.5\% of the base data.


\begin{figure}[t]
\centering
\includegraphics[scale=0.14]{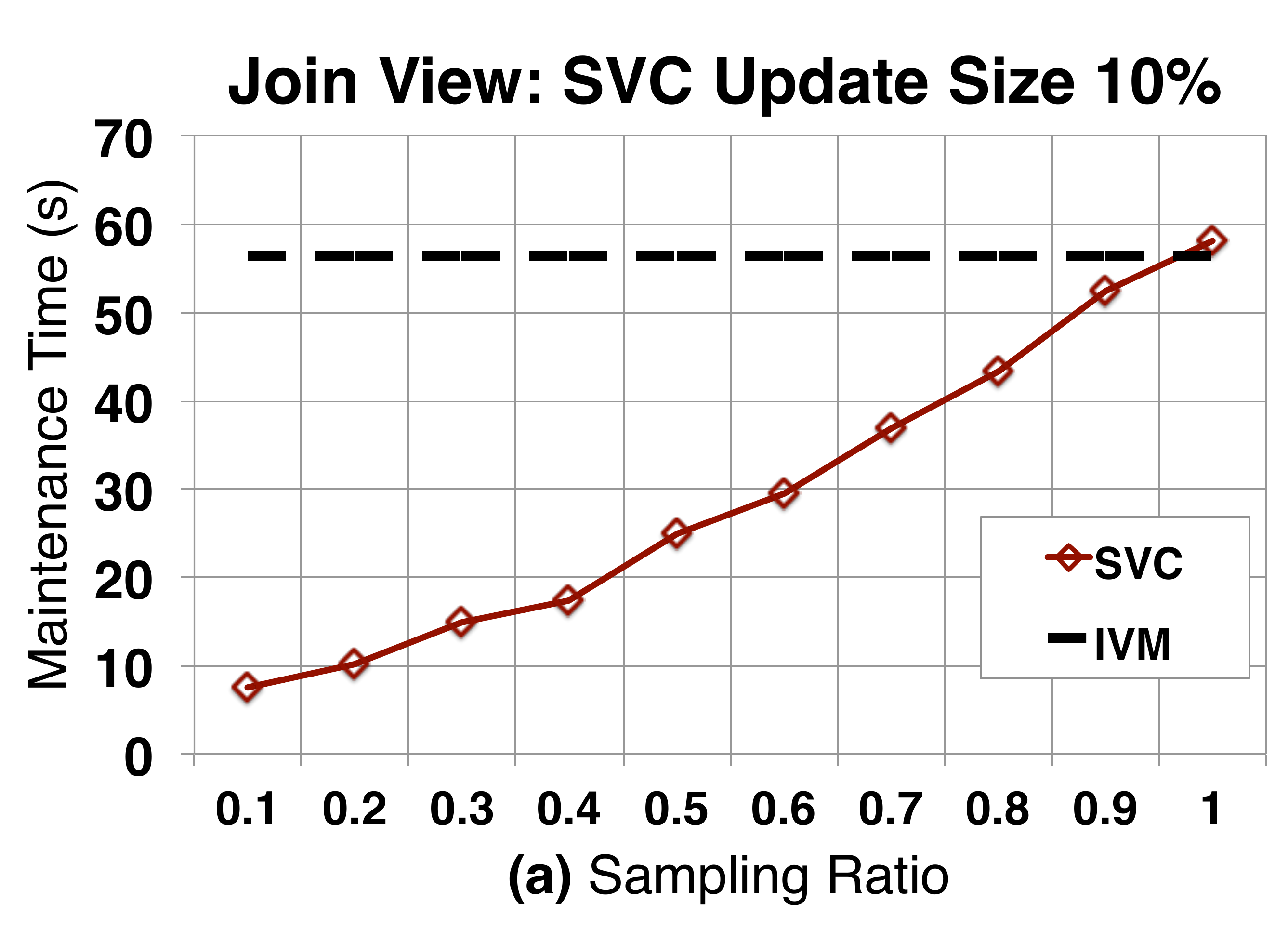}
\includegraphics[scale=0.14]{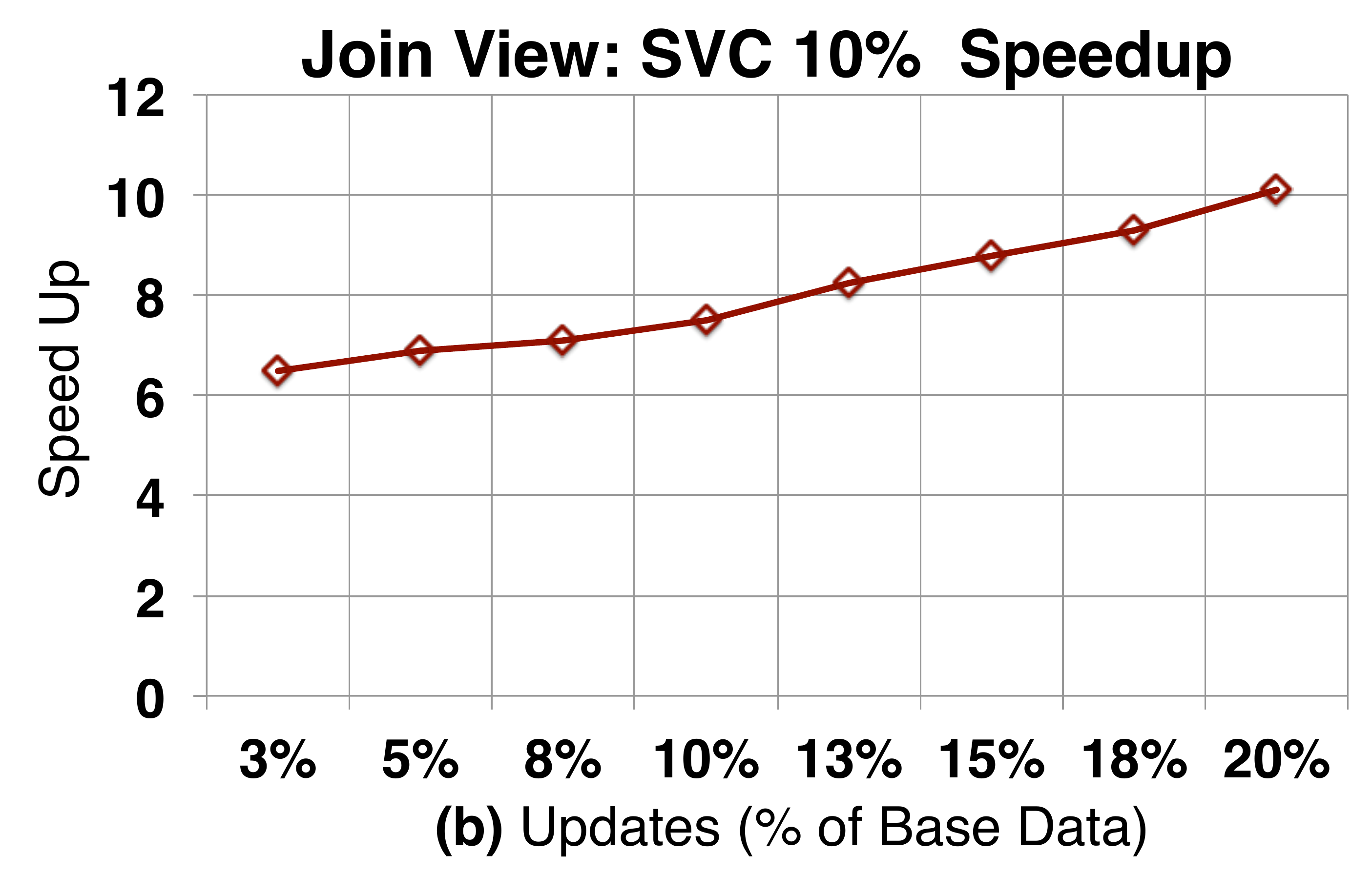}\vspace{-1em}
 \caption{(a) On a 10GB view with 1GB of insertions and updates, we vary the sampling ratio and measure the maintenance time of \svc. (b) For a fixed sampling ratio of 10\%, we vary the update size and plot the speedup compared to full incremental maintenance. \vspace{-.5em}\label{exp-1-samplesize}}\vspace{-0.5em}
\end{figure}

\begin{figure}[t]
\centering
\includegraphics[scale=0.15]{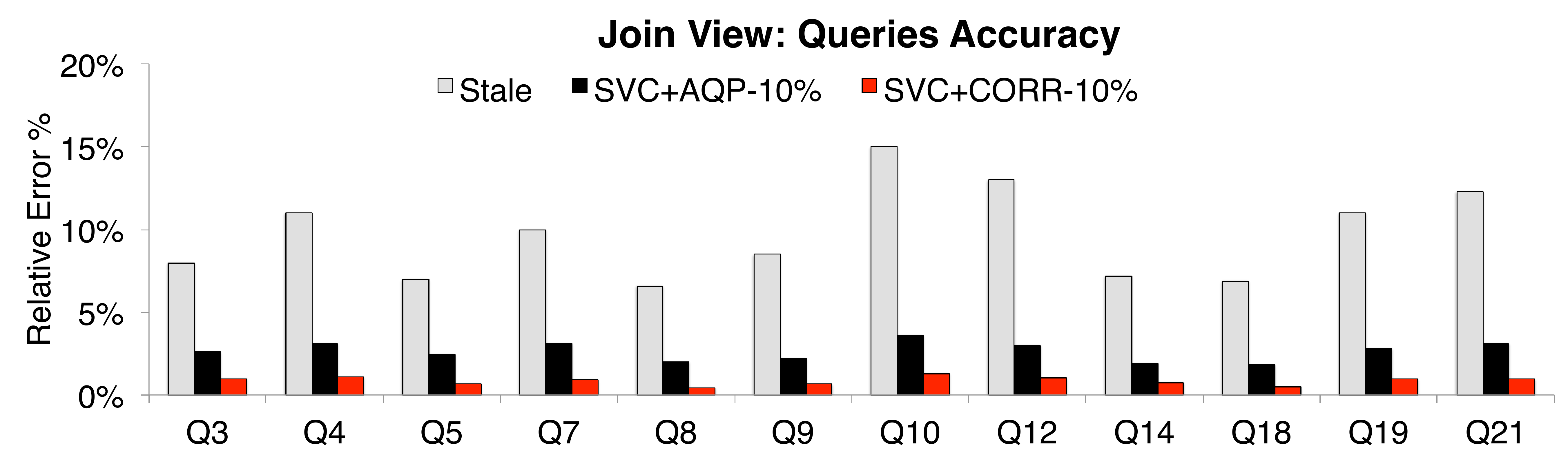}\vspace{-0.5em}
 \caption{For a fixed sampling ratio of 10\% and update size of 10\% (1GB), we generate 100 of each TPCD parameterized queries and answer the queries using the stale materialized view, \svcnospace+CORR, and \svcnospace+AQP. We plot the median relative error for each query. \label{exp-1-acc}}\vspace{-0.5em}
\end{figure}

\begin{figure}[t]
\centering
 \includegraphics[scale=0.13]{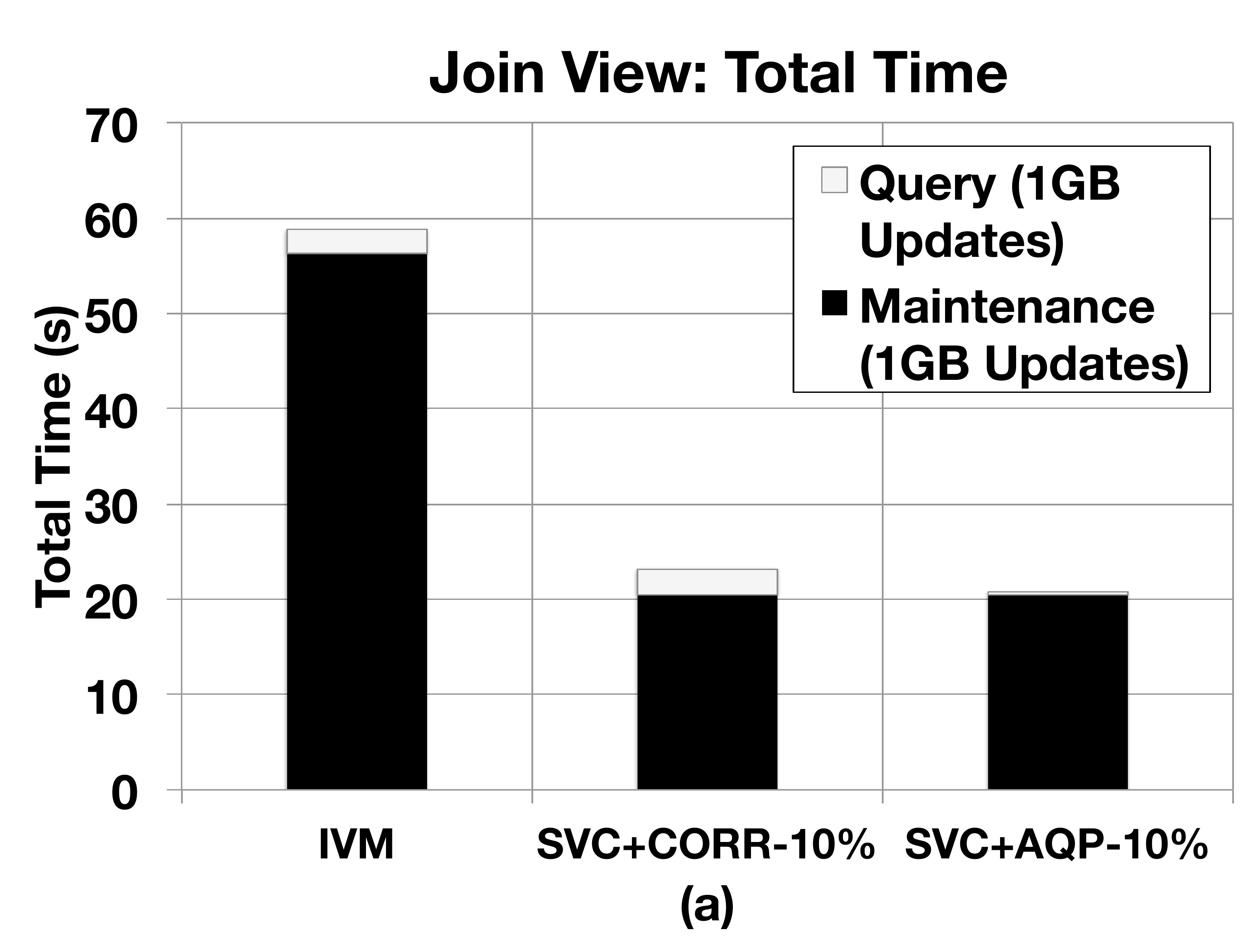}
 \includegraphics[scale=0.13]{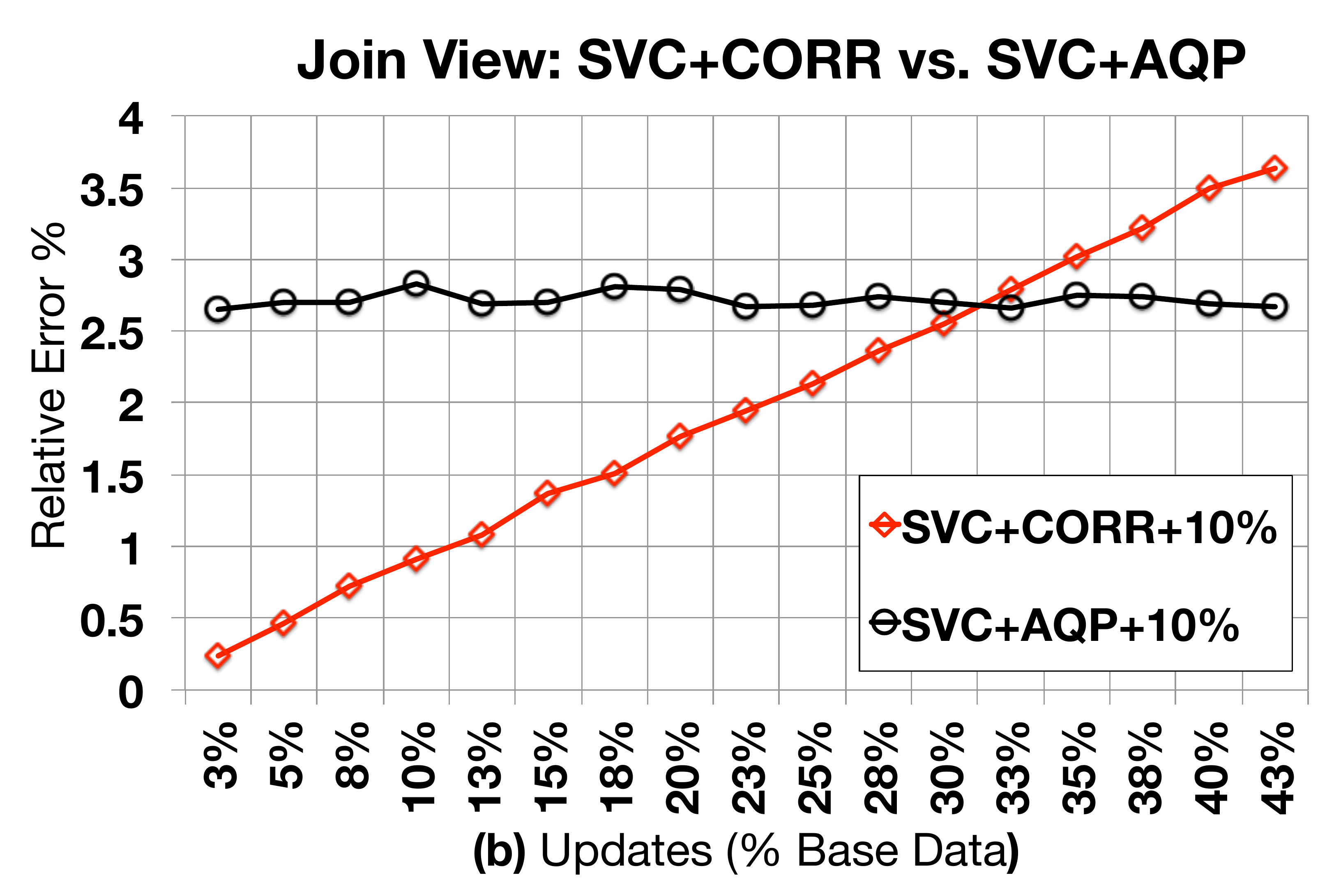}
  \caption{(a) For a fixed sampling ratio of 10\% and update size of 10\% (1GB), we measure the total time incremental maintenance + query time. (b) \svcnospace+CORR is more accurate than \svcnospace+AQP until a break even point. \vspace{-1em} \label{exp-1-total}}
\end{figure}

\subsection{Complex Views}\label{complexviews}
\begin{figure}[t]
\centering
 \includegraphics[scale=0.16]{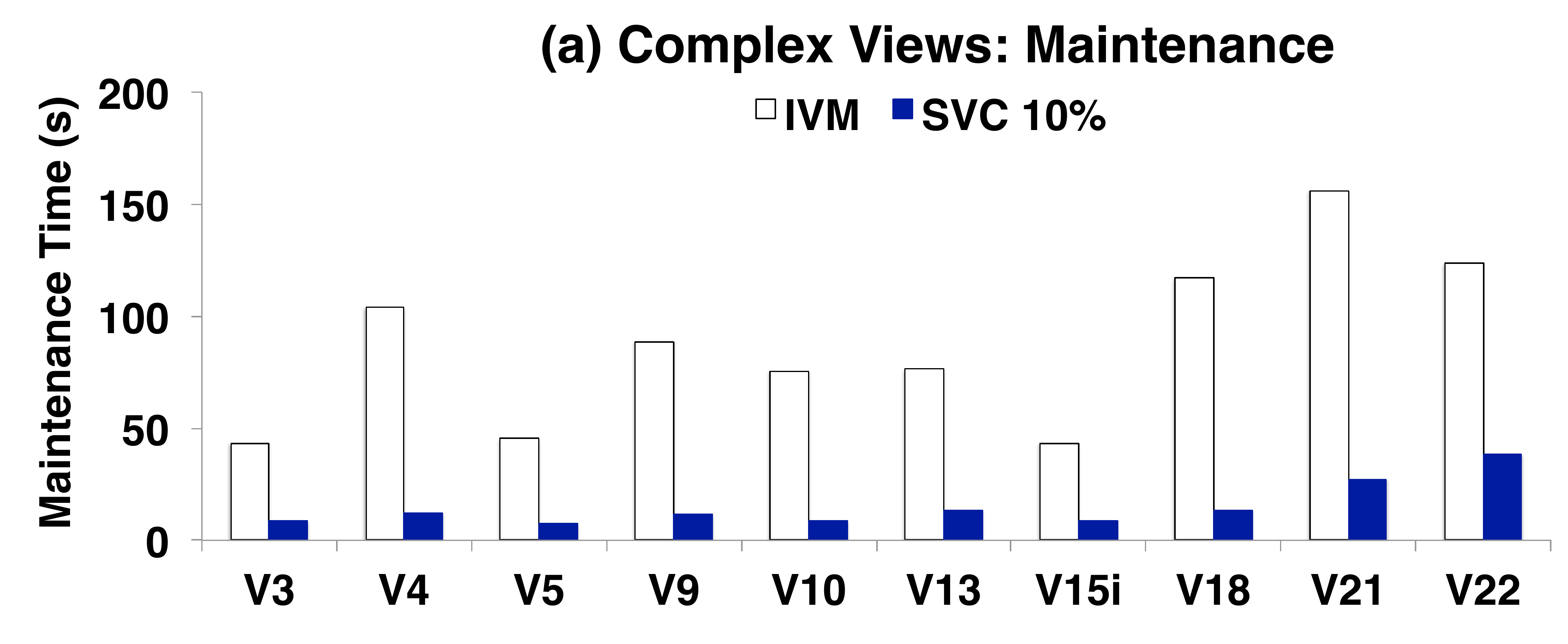}
 \includegraphics[scale=0.16]{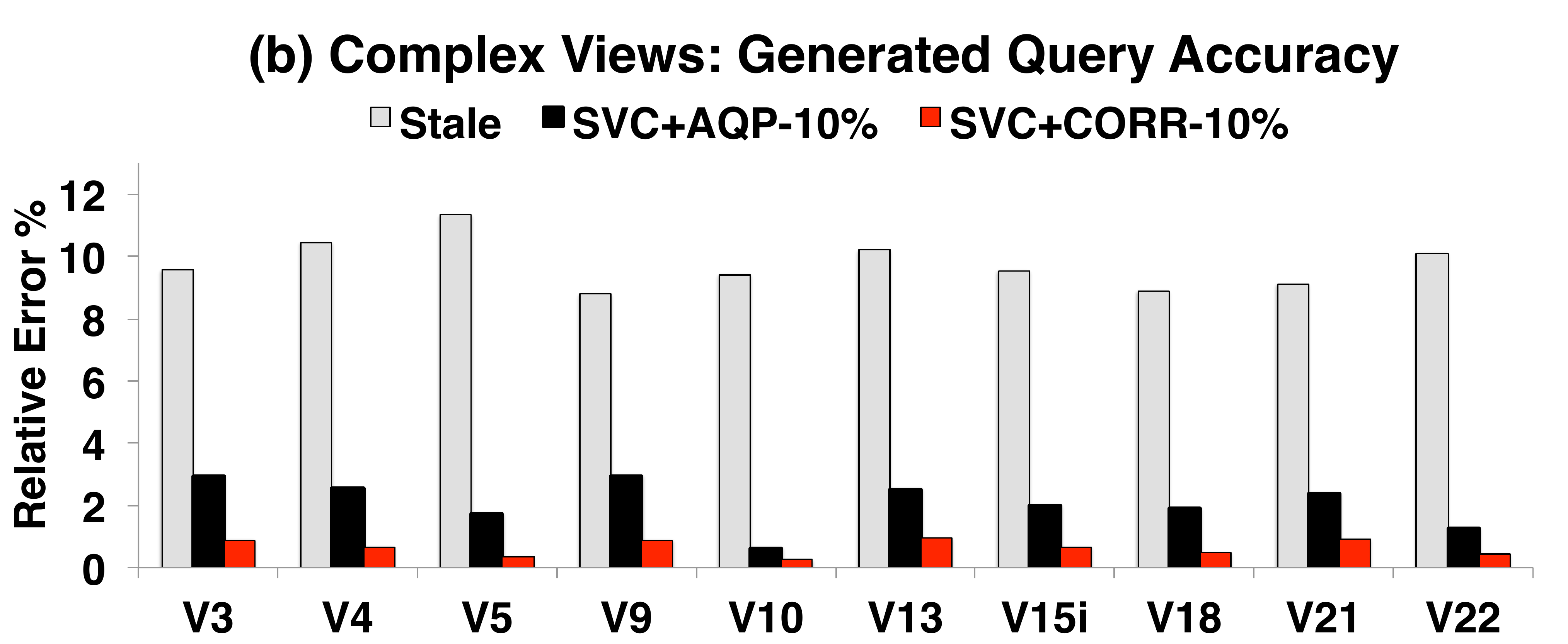}
 \caption{(a) For 1GB update size, we compare maintenance time and accuracy of \svc with a 10\% sample on different views. V21 and V22 do not benefit as much from \svc due to nested query structures. (b) For a 10\% sample size and 10\% update size, \svcnospace+CORR is more accurate than \svcnospace+AQP and No Maintenance.\label{exp3-acc}}
\end{figure}\vspace{-0.5em}
In this experiment, we demonstrate the breadth of views supported by \svc by using the TPCD queries as materialized views.
We generate a 10GB base TPCD-Skew dataset with skew $z=2$, and derive the views from this dataset.
We first generate 1GB (10\% of the base data) of updates (insertions and updates to existing records), and vary the sample size.
Figure \ref{exp3-acc} shows the maintenance time for a 10\% sample compared to the full view.
This experiment illustrates how the view definitions plays a role in the efficiency of our approach.
For the last two views, V21 and V22, we see that sampling does not lead to as large of speedup indicated in our previous experiments.  
This is because both of those views contain nested structures which block the pushdown of hashing.
V21 contains a subquery in its predicate that does not involve the primary key, but still requires a scan of the base relation to evaluate.
V22 contains a string transformation of a key blocking the push down.
These results are consistent with our previous experiments showing that \svc is faster than IVM and more accurate than \svcnospace+AQP and no maintenance.

\subsection{Outlier Indexing}

\begin{figure}[t]
\centering
 \includegraphics[scale=0.14]{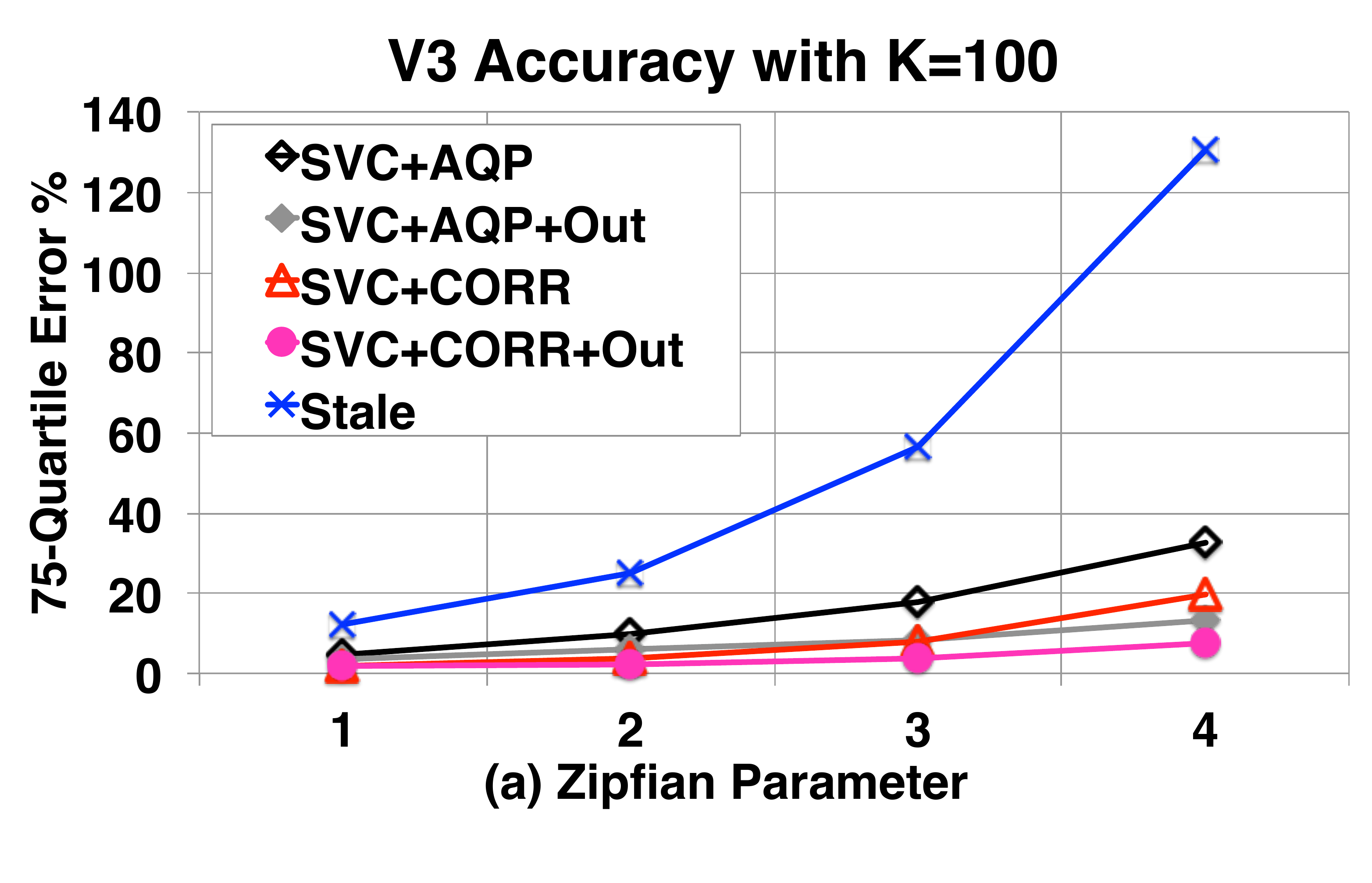}
 \includegraphics[scale=0.14]{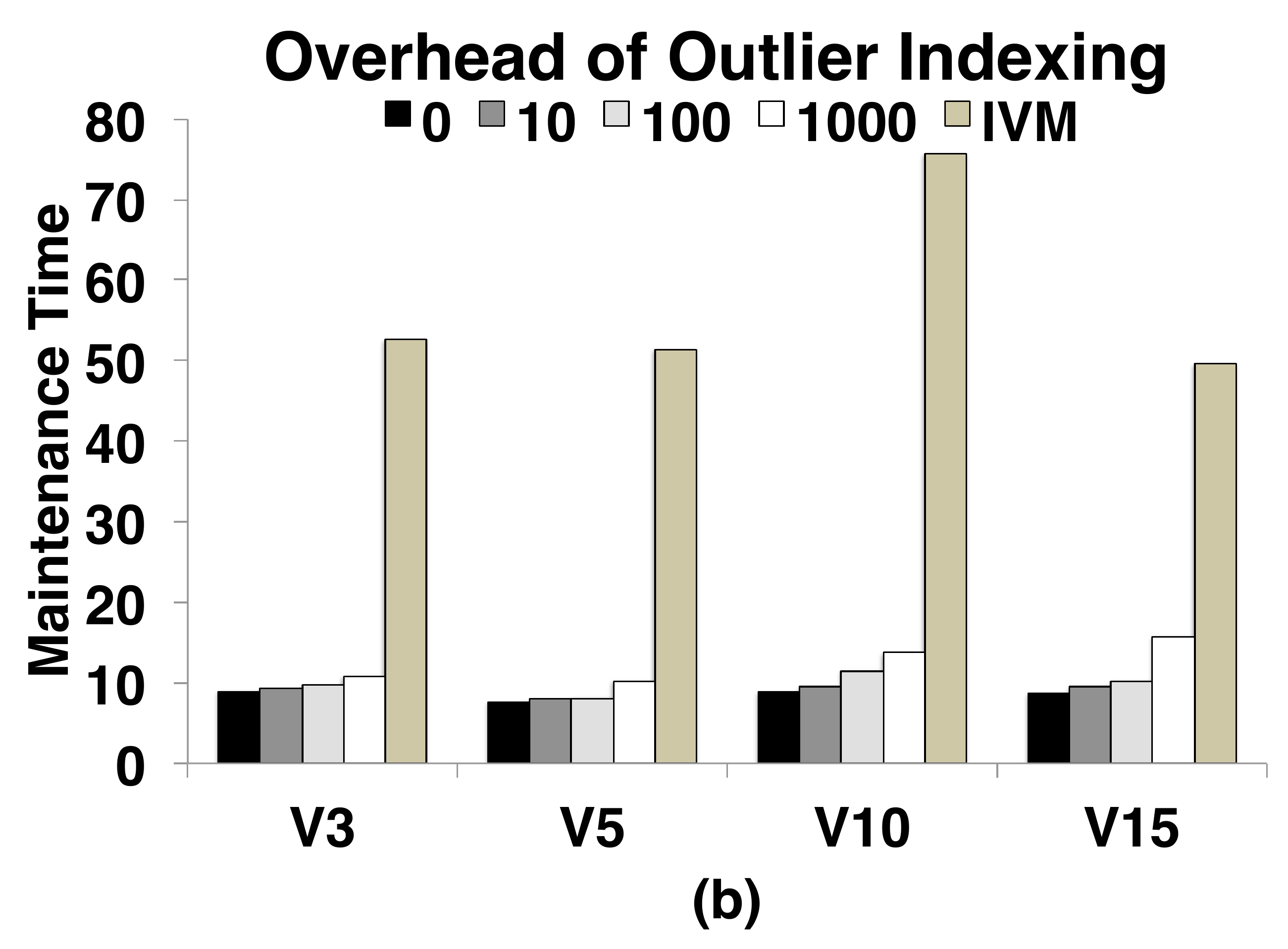}\vspace{-1em}
 \caption{(a) For one view V3 and 1GB of updates, we plot the 75\% quartile error with different techniques as we vary the skewness of the data. (b) While the outlier index adds an overhead this is small relative to the total maintenance time.\label{exp5-oi}}
\end{figure}
In our next experiment, we evaluate our outlier indexing with the top-k strategy described in Section \ref{outlier}.
In this setting, outlier indexing significantly helps for both SVC+AQP and SVC+CORR. 
We index the \textsf{l\_extendedprice} attribute in the \textsf{lineitem} table.
We evaluate the outlier index on the complex TPCD views.
We find that four views: V3, V5, V10, V15, can benefit from this index with our push-up rules. 
These are four views dependent on \textsf{l\_extendedprice} that were also in the set of ``Complex'' views chosen before.

\vspace{1em}

In our first outlier indexing experiment (Figure \ref{exp5-oi}(a)), we analyze V3.
We set an index of 100 records, and applied \svcnospace+CORR and \svcnospace+AQP to views derived from a dataset with a skew parameter $z=\{1,2,3,4\}$. 
We run the same queries as before, but this time we measure the error at the 75\% quartile.
We find in the most skewed data \svc with outlier indexing reduces query error by a factor of 2.
Next, in Figure \ref{exp5-oi} (b), we plot the overhead for outlier indexing for V3 with an index size of 0, 10, 100, and 1000.
While there is an overhead, it is still small compared to the gains made by sampling the maintenance strategy.
We note that none of the prior experiments used an outlier index. 
The caveat is that these experiments were done with moderately skewed data with Zipfian parameter = 2, if this parameter is set to 4 then the 75\% quartile query estimation error is nearly 20\% (Figure \ref{exp5-oi}a).
Outlier indexing always improves query results as we are reducing the variance of the estimation set, however, this reduction in variance is largest when there is a longer tail.

\subsection{Conviva}
\begin{figure}[t]
\centering
 \includegraphics[scale=0.16]{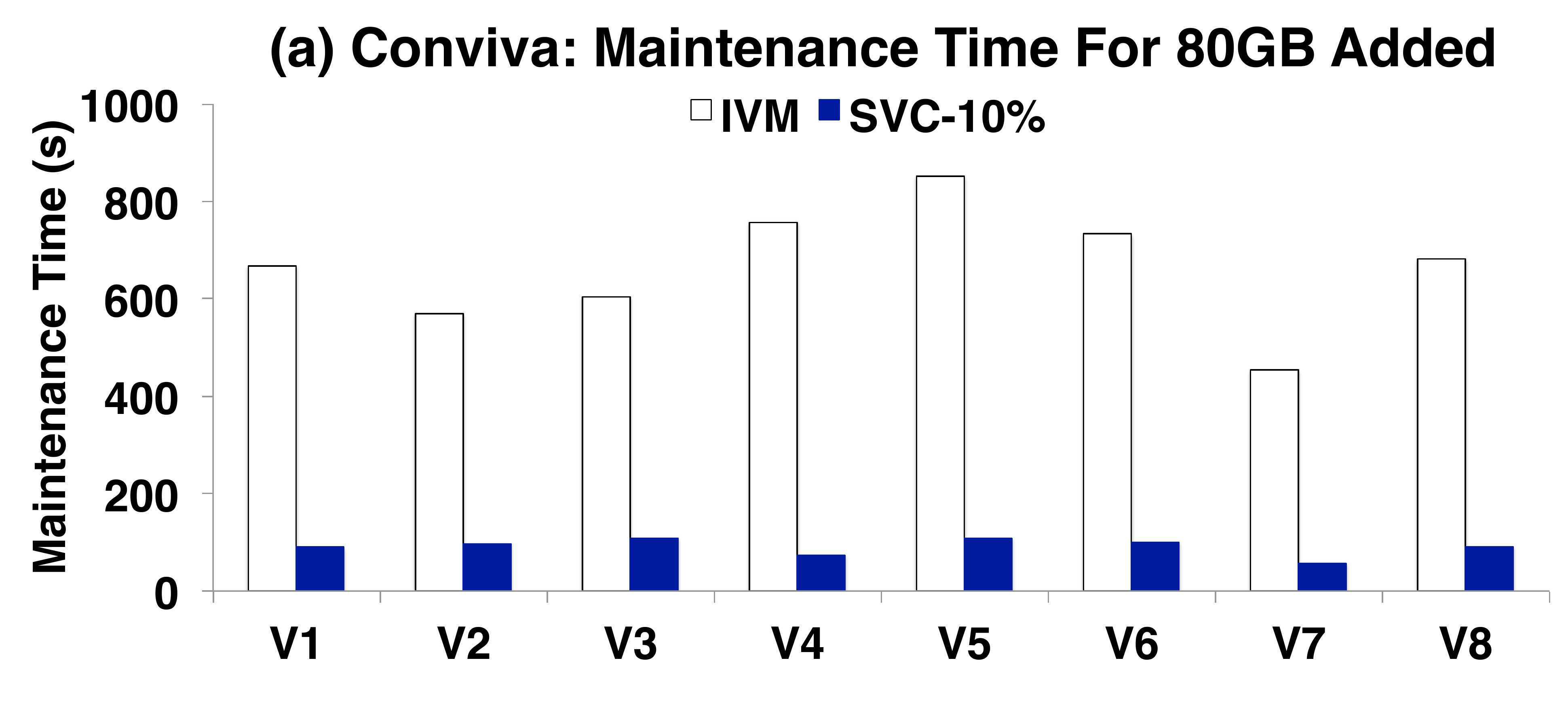}
 \includegraphics[scale=0.16]{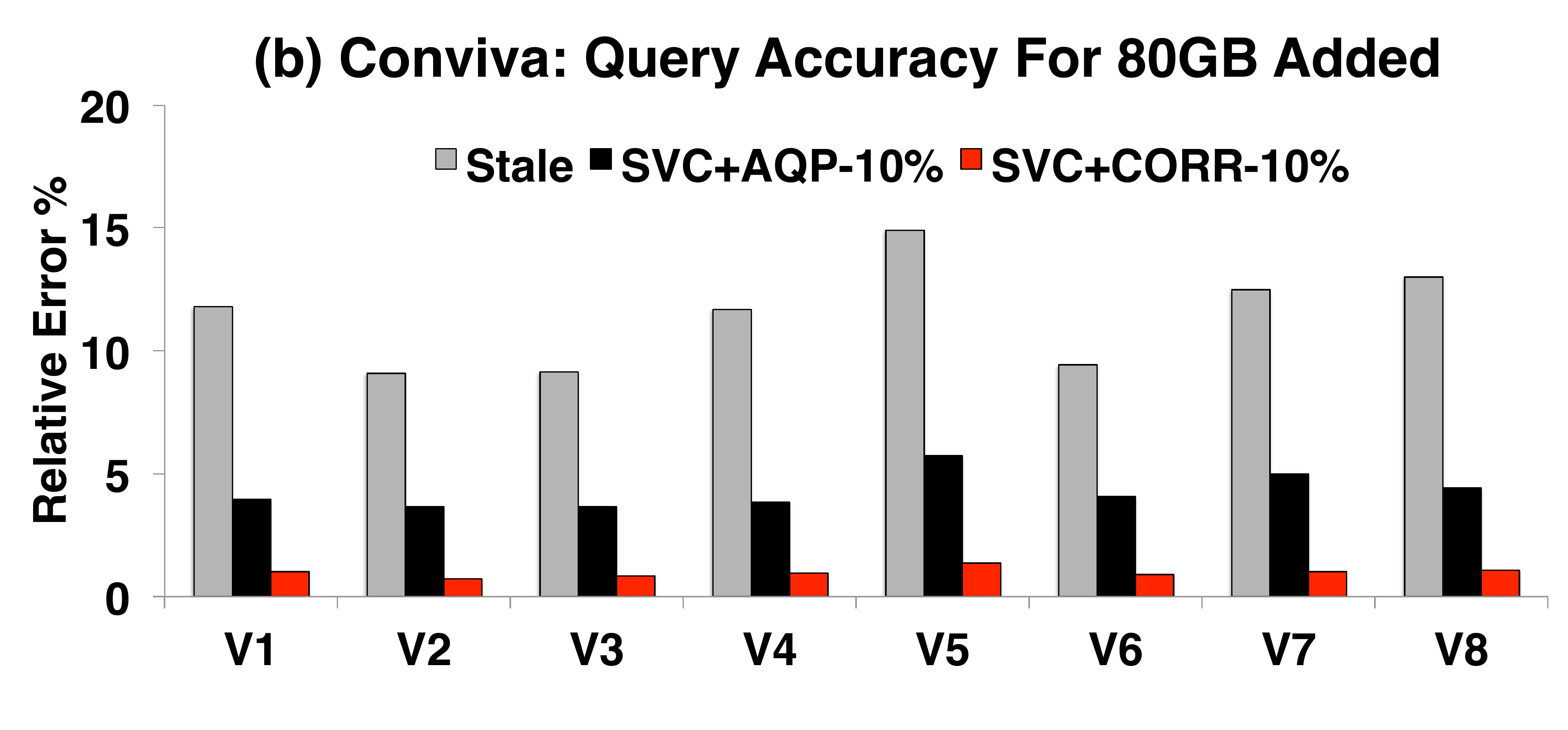} \vspace{-1.5em}
 \caption{(a) We compare the maintenance time of \svc with a 10\% sample and full incremental maintenance, and find that as with TPCD \svc saves significant maintenance time. (b) We also evaluate the accuracy of the estimation techniques. \label{conv-1}}\vspace{-1.5em}
\end{figure}
We derive the views from 800GB of base data and add 80GB of updates. 
These views are stored and maintained using Apache Spark in a distributed environment.
The goal of this experiment is to evaluate how \svc performs in a real world scenario with a real dataset and a distributed architecture.
In Figure \ref{conv-1}(a), we show that on average over all the views, \svc-10\% gives a 7.5x speedup.
For one of the views full incremental maintenance takes nearly 800 seconds, even on a 10-node cluster, which is a very significant cost.
In Figure \ref{conv-1}(b), we show that \svc also gives highly accurate results with an average error of 0.98\%.
These results show consistency with our results on the synthetic datasets.
This experiment highlights a few salient benefits of \svc: (1) sampling is a relatively cheap operation and the relative speedups in a single node and distributed environment are similar, (2) for analytic workloads like Conviva (i.e., user engagement analysis) a 10\% sample gives results with 99\% accuracy, and (3) savings are still significant in systems like Spark that do not support selective updates.

\subsection{Additional Experiments}

\subsubsection{Aggregate View}
\label{exp-datacube}

\begin{figure}[t]
\centering
 \includegraphics[scale=0.15]{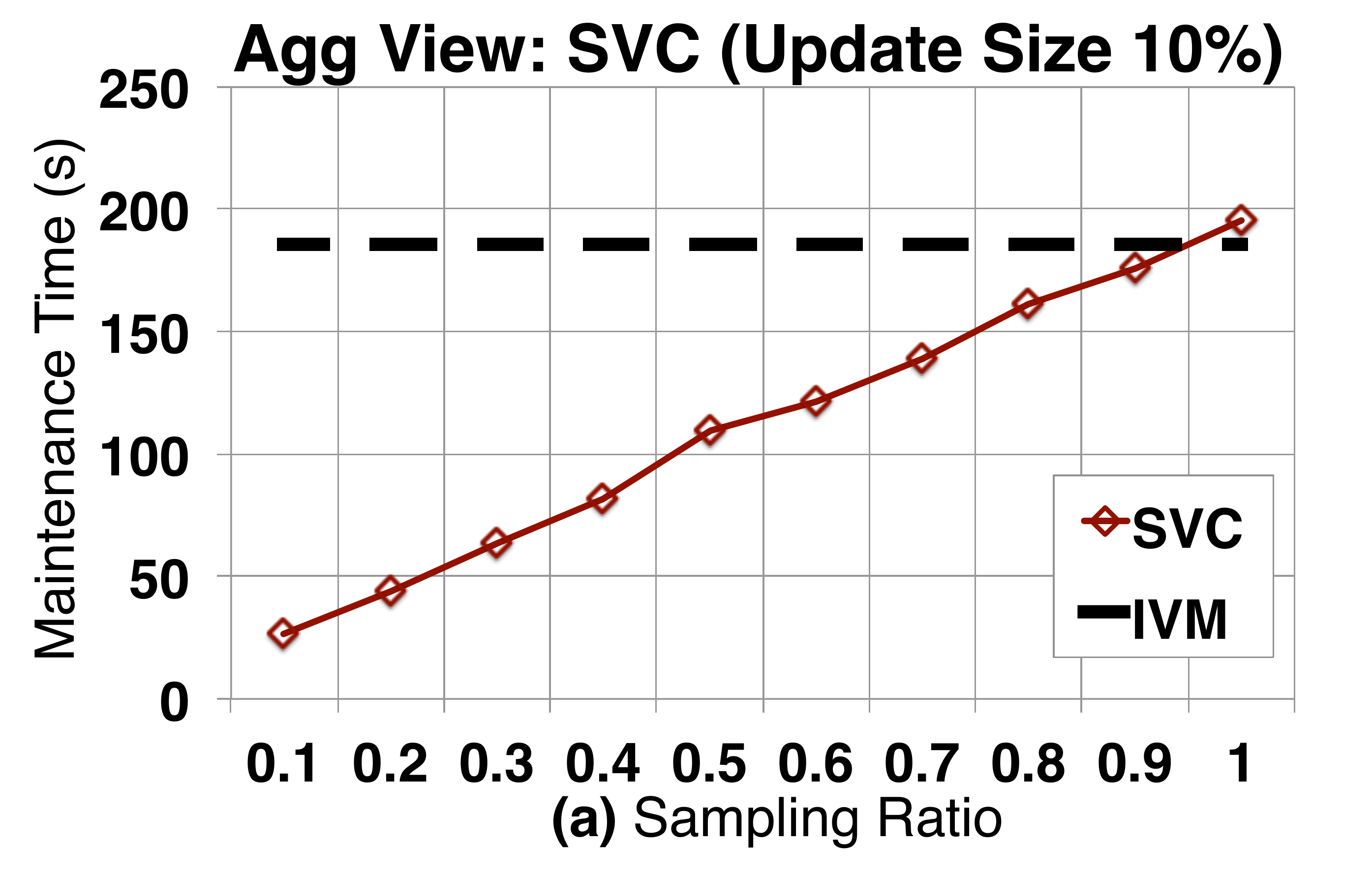}
 \includegraphics[scale=0.15]{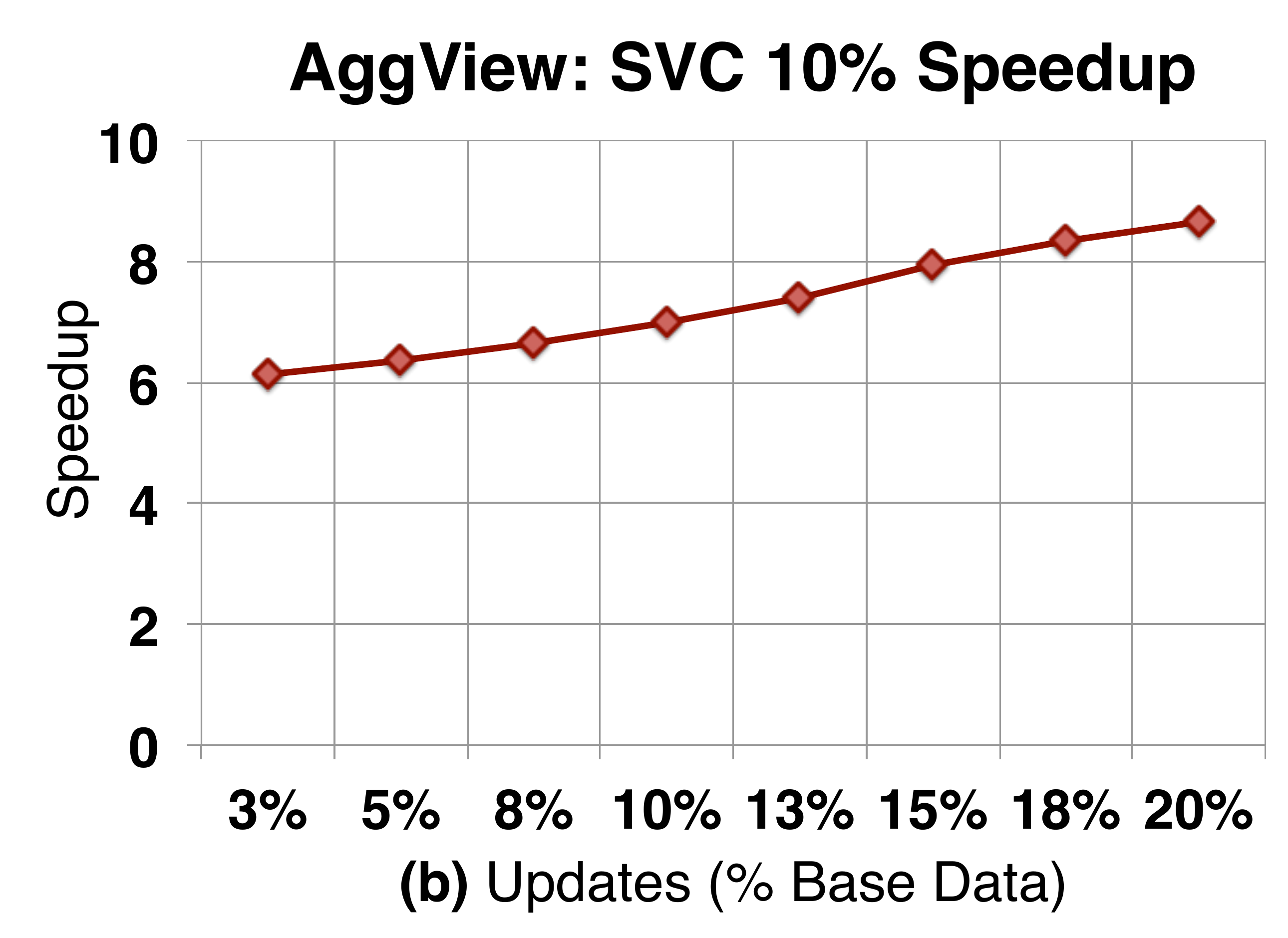}
   \caption{(a) In the aggregate view case, sampling can save significant maintenance time. (b) As the update size grows SVC tends towards an ideal speedup of 10x.\label{exp2-acc-sample}}
\end{figure}

\begin{figure}[t]
\centering
 \includegraphics[scale=0.17]{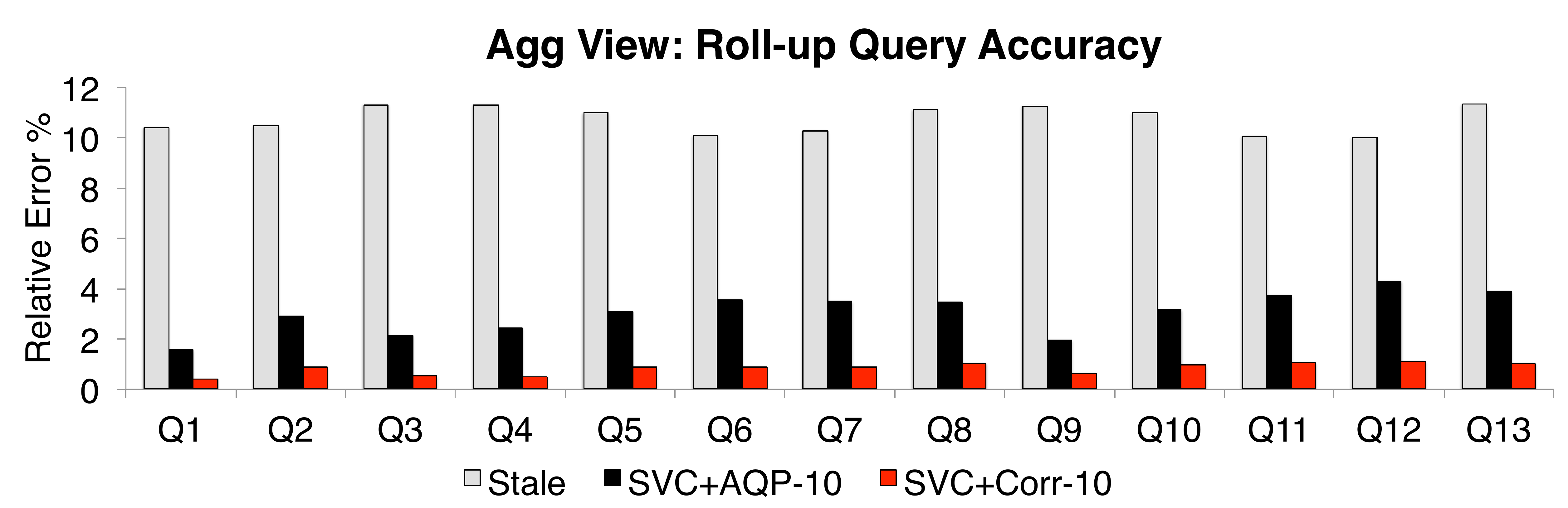}
   \caption{We measure the accuracy of each of the roll-up aggregate queries on this view. For a 10\% sample size and 10\% update size, we find that SVC+Corr is more accurate than SVC+AQP and No Maintenance.\label{exp2-acc-sample2}}
\end{figure}

\begin{figure}[t]
\centering
\includegraphics[scale=0.17]{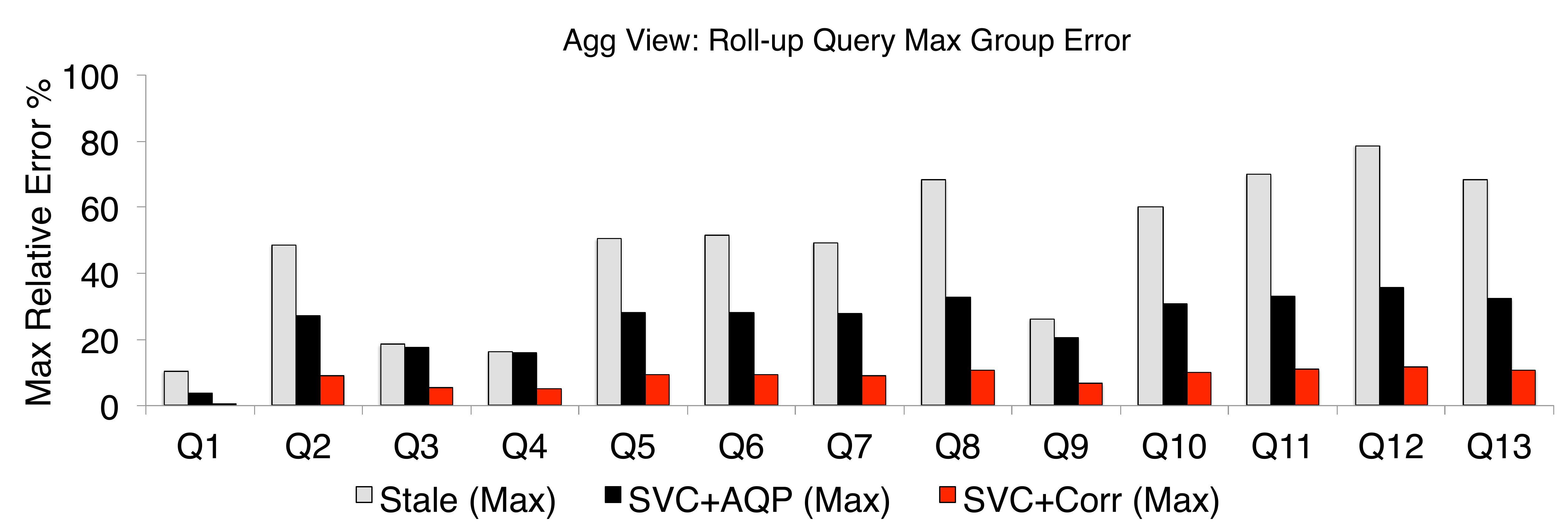}
   \caption{For 1GB of updates, we plot the max error as opposed to the median error in the previous experiments. Even though updates are 10\% of the dataset size, some queries are nearly 80\% incorrect. SVC helps significantly mitigate this error. \label{exp2-max}}
\end{figure}

\begin{figure}[t]
\centering
  \includegraphics[scale=0.17]{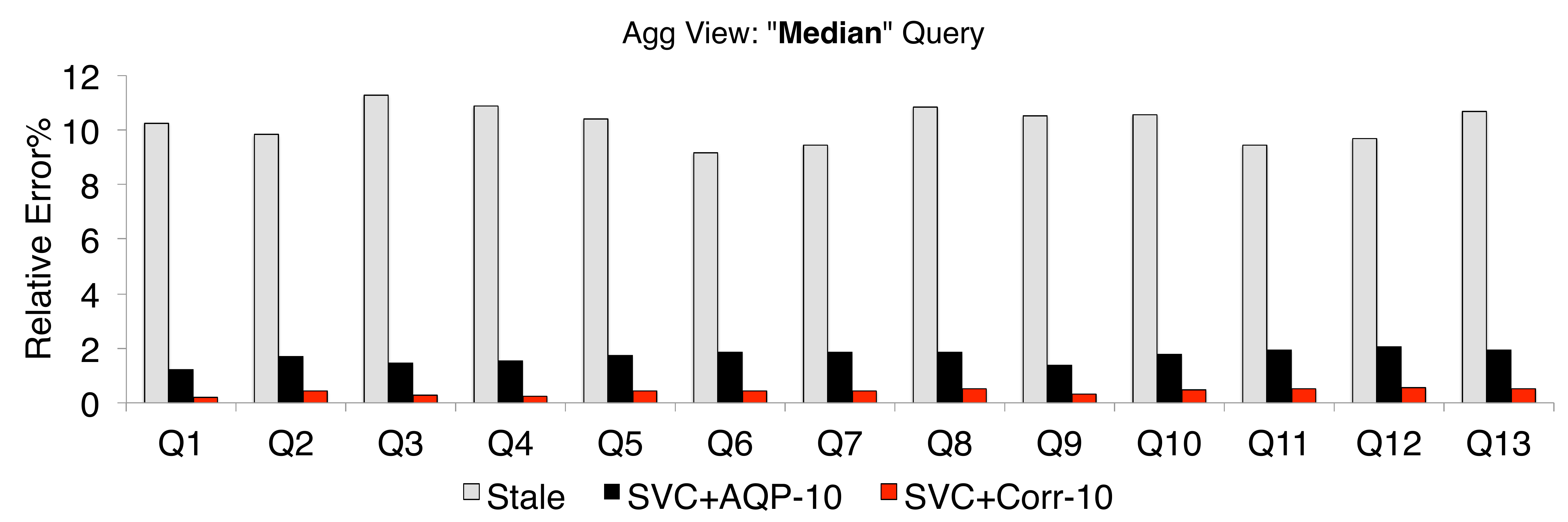}
 \caption{We run the same experiment but replace the \sumfunc query with a median query. We find that similarly SVC is more accurate.\label{exp2-median} }
\end{figure}

In our next experiment, we evaluate an aggregate view use case similar to a data cube.
We generate a 10GB base TPCD dataset with skew $z=1$, and derive the base cube as a materialized view from this dataset.
We add 1GB of updates and apply SVC to estimate the results of all of the ``roll-up" dimensions.

\textbf{Performance: }
We observed the same trade-off as the previous experiment where sampling significantly reduces the maintenance time (Figure \ref{exp2-acc-sample}(a)).
It takes 186 seconds to maintain the entire view, but a 10\% sample can be maintained in 26 seconds.
As before, we fix the sample size at 10\% and vary the update size.
We similarly observe that SVC becomes more efficient as the update size grows (Figure \ref{exp2-acc-sample}(b)), and at an update size of 20\%  the speedup is 8.7x.

\textbf{Accuracy: }
In Figure \ref{exp2-acc-sample2}, we measure the accuracy of each of the ``roll-up" aggregate queries on this view.
That is, we take each dimension and aggregate over the dimension.
We fix the sample size at 10\% and the update size at 10\%.
On average SVC+Corr is 12.9x more accurate than the stale baseline and 3.6x more accurate than SVC+AQP (Figure \ref{exp2-acc-sample}(c)). 

Since the data cubing operation is primarily constructed by group-by aggregates, we can also measure the max error for each of the aggregates.
We see that while the median staleness is close to 10\%, for some queries some of the group aggregates have nearly 80\% error (Figure \ref{exp2-max}).
SVC greatly mitigates this error to less than 12\% for all queries.

\textbf{Other Queries: }
Finally, we also use the data cube to illustrate how SVC can support a broader range of queries outside of \sumfunc, \countfunc, and \avgfunc.
We change all of the roll-up queries to use the \textbf{median} function (Figure \ref{exp2-median}).
First, both SVC+Corr and SVC+AQP are more accurate as estimating the median than they were for estimating sums. 
This is because the median is less sensitive to variance in the data.

\subsubsection{Mini-batch Experiments}
\begin{figure}[t]
\centering
 \includegraphics[scale=0.14]{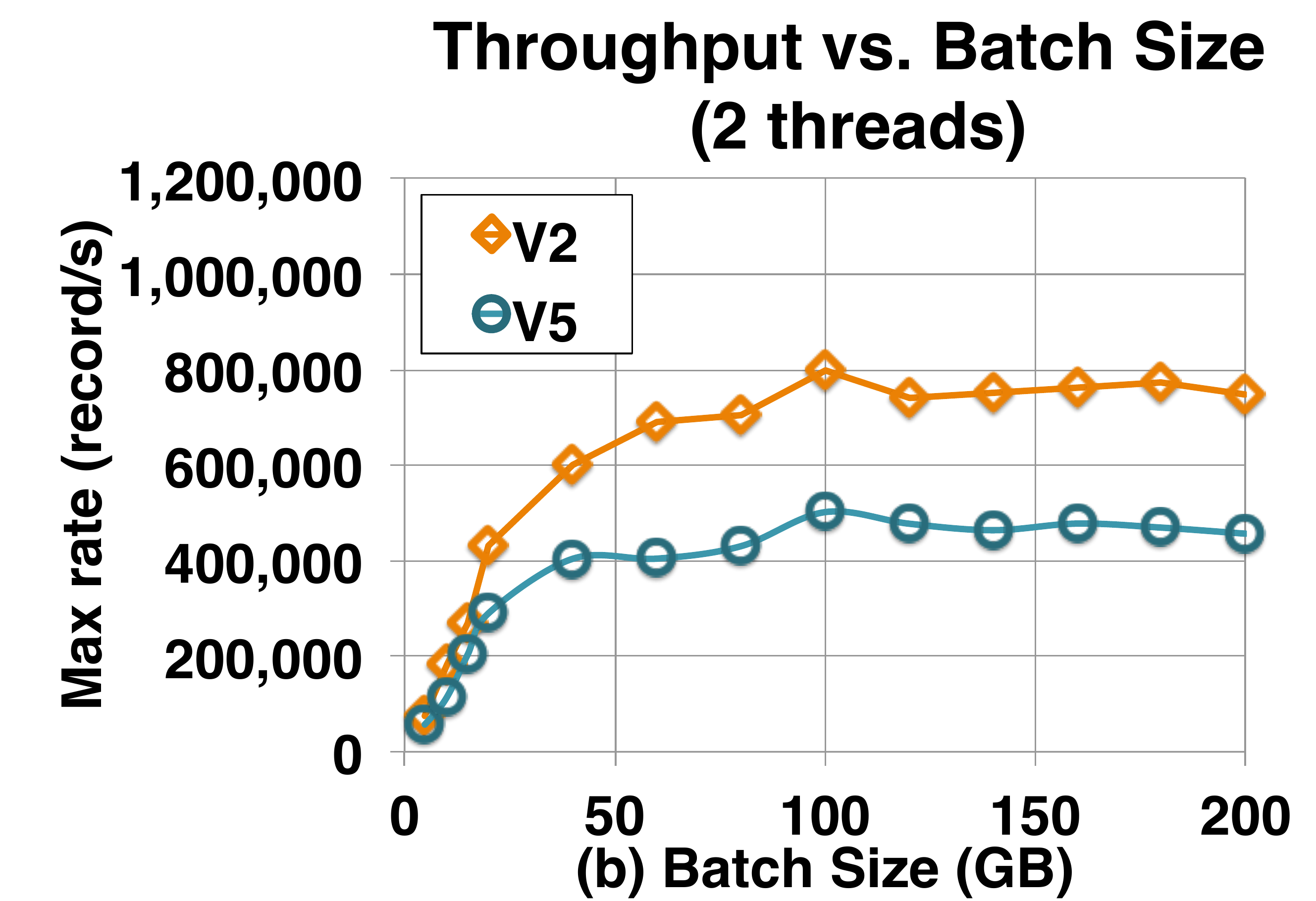}
 \includegraphics[scale=0.14]{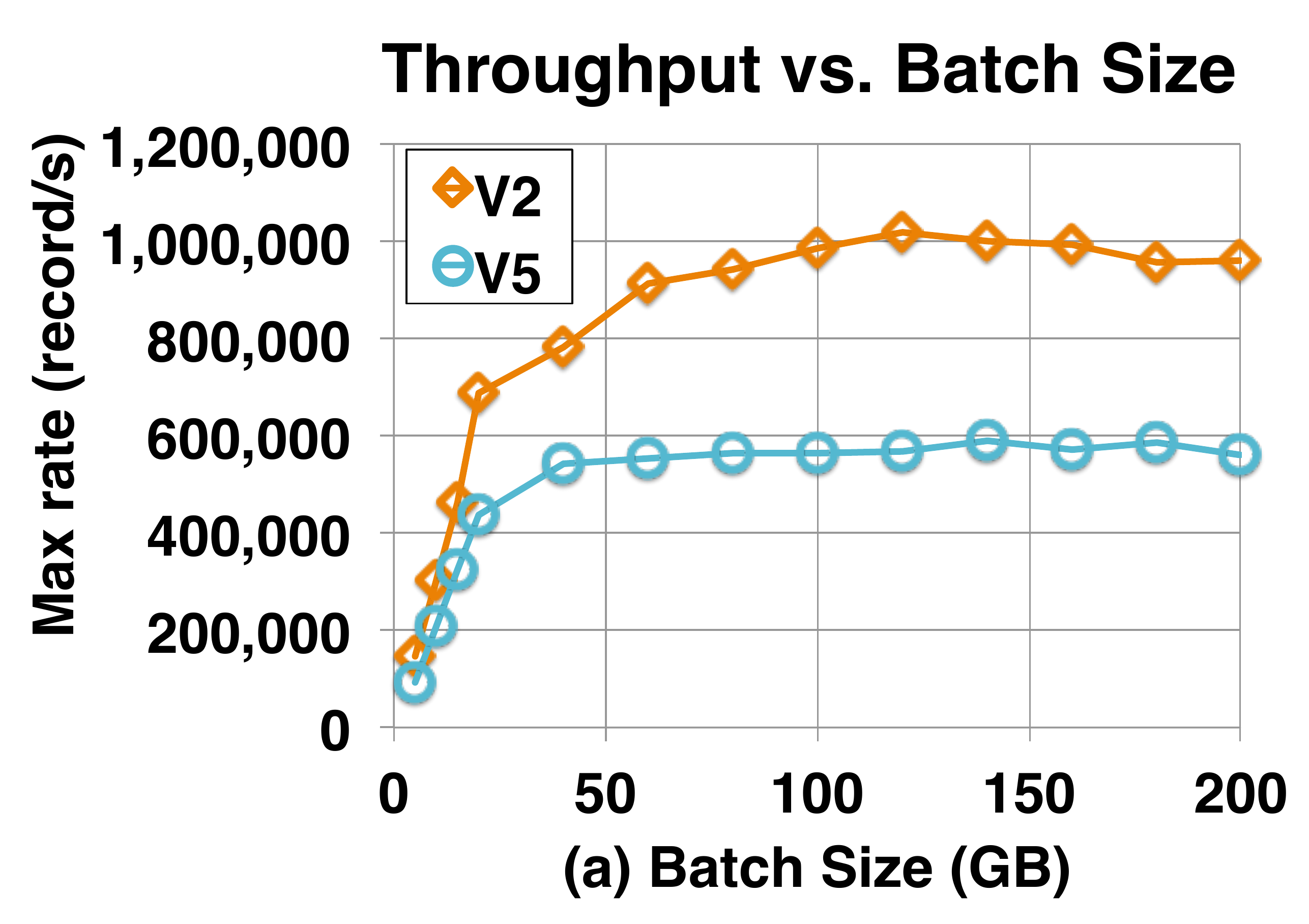}
 \caption{(a) Spark RDDs are most efficient when updated in batches. As batch sizes increase the system throughput increases. (b) When running multiple threads, the throughput reduces. However, larger batches are less affected by this reduction. \label{conv-2}}
\end{figure}

We devised an end-to-end experiment simulating a real integration with periodic maintenance.
However, unlike the MySQL case, Apache Spark does not support selective updates and insertions as the ``views" are immutable.
A further point is that the immutability of these views and Spark's fault-tolerance requires that the ``views" are maintained synchronously.
Thus, to avoid these significant overheads, we have to update these views in batches.
Spark does have a streaming variant \cite{zaharia2012discretized}, however, this does not support the complex SQL derived materialized views used in this paper, and still relies on mini-batch updates.

SVC and IVM will run in separate threads each with their own RDD materialized view.
In this application, both SVC and IVM maintain respective their RDDs with batch updates.
In this model, there are a lot of different parameters: batch size for periodic maintenance, batch size for SVC, sampling ratio for SVC, and the fact that concurrent threads may reduce overall throughput.
Our goal is to fix the throughput of the cluster, and then measure whether SVC+IVM or IVM alone leads to more accurate query answers.

\textbf{Batch sizes:} In Spark, larger batch sizes amortize overheads better.
In Figure \ref{conv-2}(a), we show a trade-off between batch size and throughput of Spark for V2 and V5.
Throughputs for small batches are nearly 10x smaller than the throughputs for the larger batches. 

\textbf{Concurrent SVC and IVM:} Next, we measure the reduction in throughput when running multiple threads.
We run SVC-10 in loop in one thread and IVM in another.
We measure the reduction in throughput for the cluster from the previous batch size experiment.
In Figure \ref{conv-2}(b), we plot the throughput against batch size when two maintenance threads are running.
While for small batch sizes the throughput of the cluster is reduced by nearly a factor of 2, for larger sizes the reduction is
smaller.
As we found in later experiments (Figure \ref{conv-5}), larger batch sizes are more amenable to parallel computation since there was more idle CPU time.

\textbf{Choosing a Batch Size:}
The results in Figure \ref{conv-2}(a) and Figure \ref{conv-2}(b) show that larger batch sizes are more efficient, however, larger batch sizes also lead to more staleness.
Combining the results in Figure \ref{conv-2}(a) and Figure \ref{conv-2}(b), for both SVC+IVM and IVM, we get cluster throughput as a function of batch size.
For a fixed throughput, we want to find the smallest batch size that achieves that throughput for both.
For V2, we fixed this at 700,000 records/sec and for V5 this was 500,000 records/sec.
For IVM alone the smallest batch size that met this throughput demand was 40GB for both V2 and V5.
And for SVC+IVM, the smallest batch size was 80GB for V2 and 100GB for V5. 
When running periodic maintenance alone view updates can be more frequent, and when run in conjunction with SVC it is less frequent. 

We run both of these approaches in a continuous loop, SVC+IVM and IVM, and measure their maximal error during a maintenance period.
There is further a trade-off with the sampling ratio, larger samples give more accurate estimates however between SVC batches they go stale.
We quantify the error in these approaches with the max error; that is the maximum error in a maintenance period (Figure \ref{conv-4}).
These competing objective lead to an optimal sampling ratio of 3\% for V2 and 6\% for V5.
At this sampling point, we find that applying SVC gives results 2.8x more accurate for V2 and 2x more accurate for V5.

\begin{figure}[t]
\centering
 \includegraphics[scale=0.14]{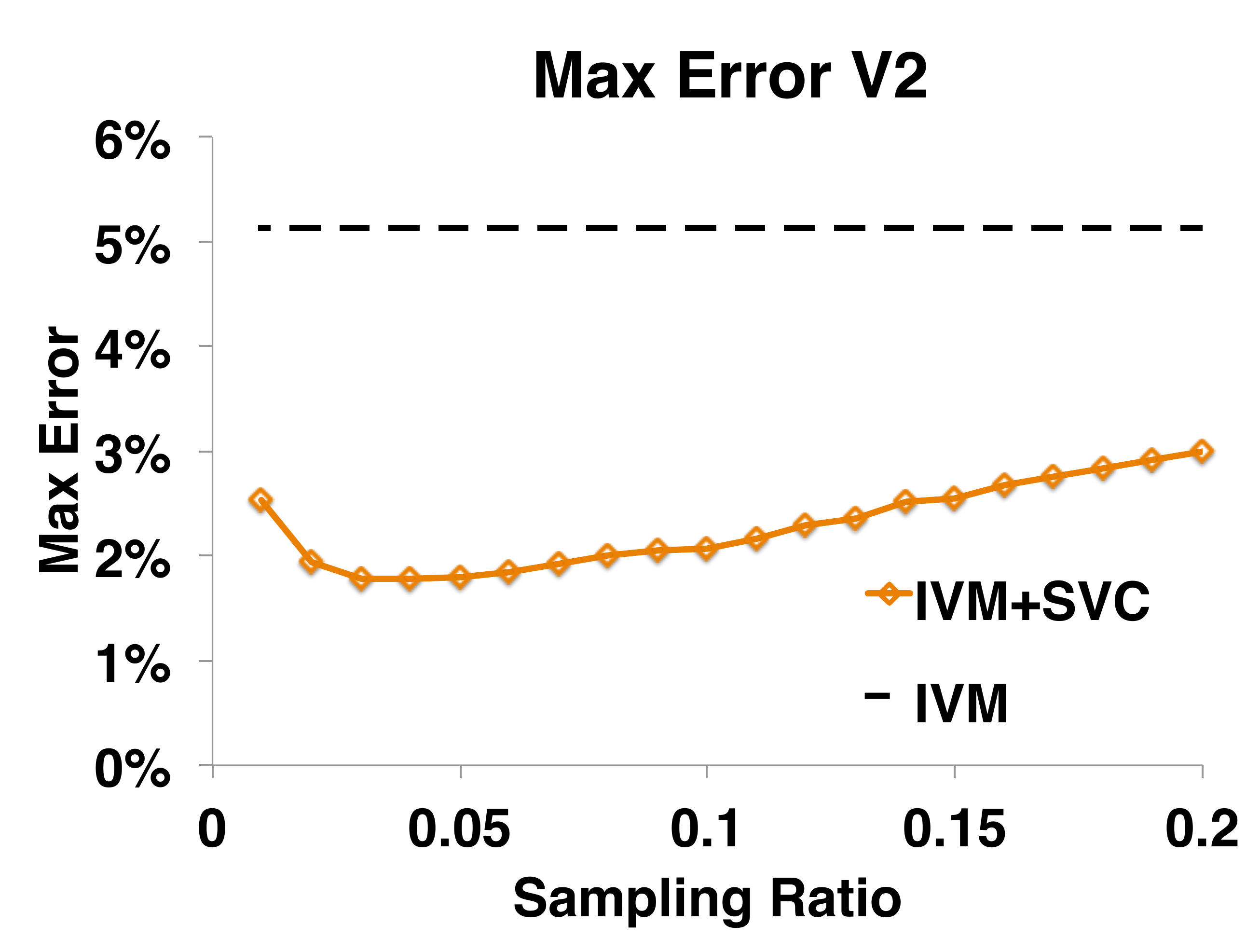}
 \includegraphics[scale=0.14]{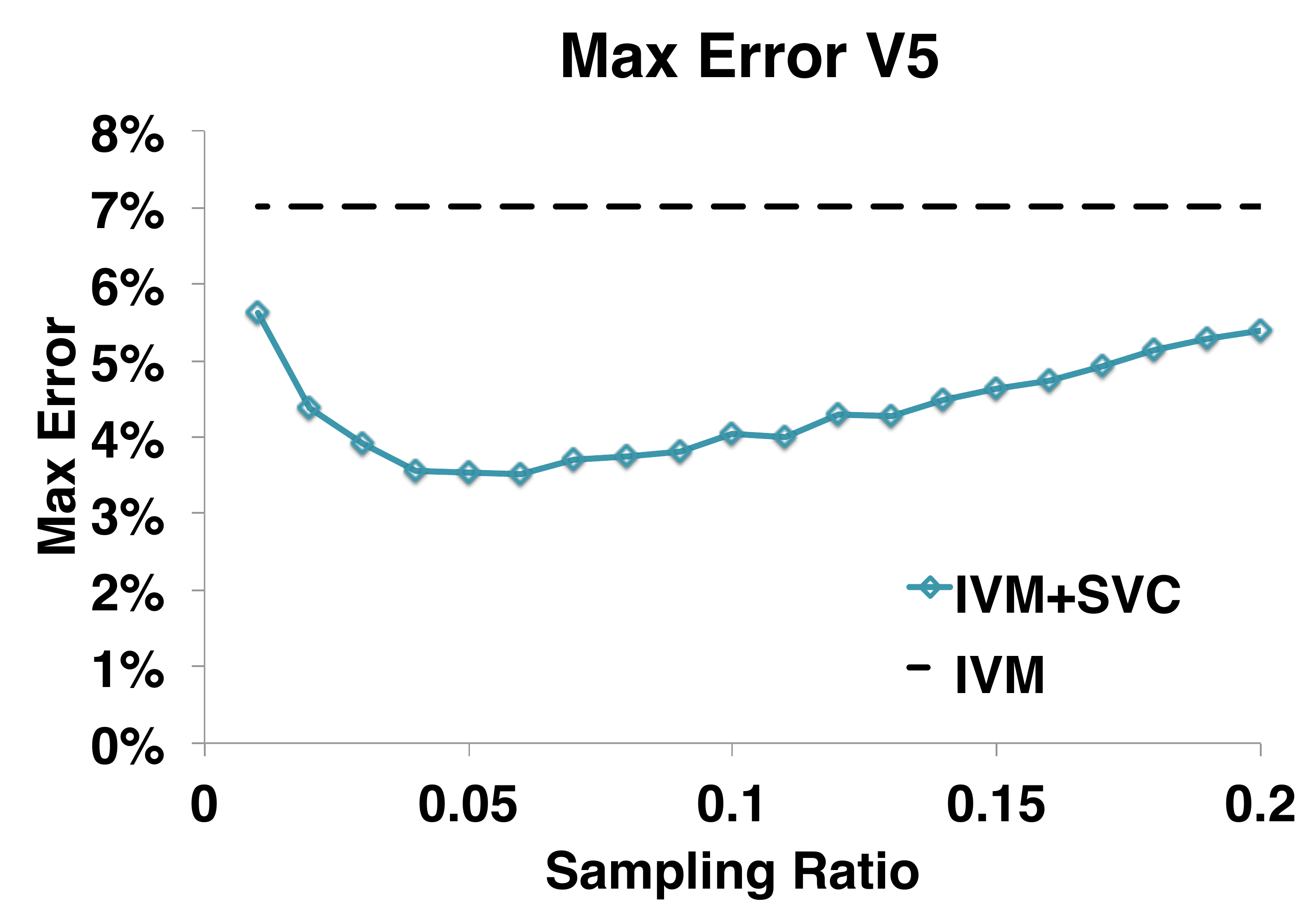}
 \caption{For a fixed throughput, SVC+Periodic Maintenance gives more accurate results for V2 and V5. \label{conv-4}} 
\end{figure}

\begin{figure}[t]
\centering
\includegraphics[width=\columnwidth]{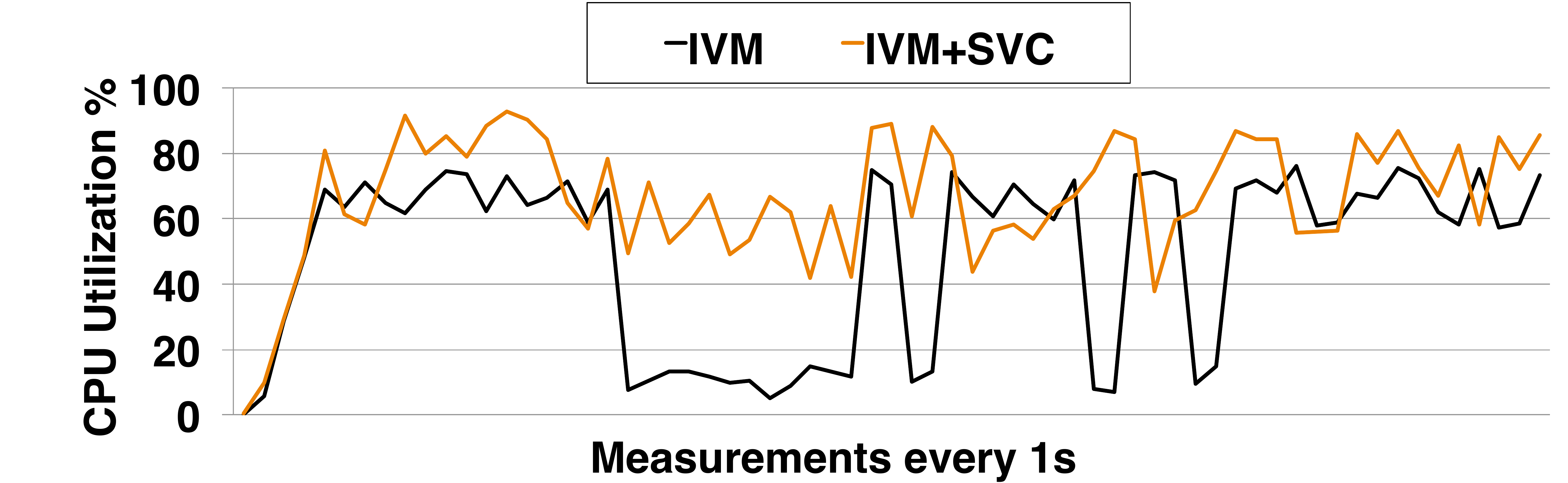}
 \caption{SVC better utilizes idle times in the cluster by maintaining the sample.\label{conv-5}} 
\end{figure}
To give some intuition on why SVC gives more accurate results, in Figure \ref{conv-5}, we plot the average CPU utilization of the cluster for both periodic IVM and SVC+periodic IVM. 
We find that SVC takes advantage of the idle times in the system; which are common during shuffle operations in a synchronous parallelism model.

In a way, these experiments present a worst-case application for SVC, yet it still gives improvements in terms of query accuracy.
In many typical deployments throughput demands are variable forcing maintenance periods to be longer, e.g., nightly.
The same way that SVC takes advantage of micro idle times during communication steps, it can provide large gains during controlled idle times when no maintenance is going on concurrently.

\vspace{-0.75em}
\section{Related Work}\label{related}
Addressing the cost of materialized view maintenance is the subject of many recent papers, which
focus on various perspectives including complex analytical queries~\cite{nikolic2014linview}, transactions~\cite{bailis2014scalable}, real-time analytics \cite{liarou2012monetdb}, and physical design~\cite{lefevre2014opportunistic}.
The increased research focus parallels a major concern in industrial systems for incrementally updating pre-computed results and indices such as Google Percolator~\cite{percolator} and Twitter's Rainbird~\cite{rainbird}.
The streaming community has also studied the view maintenance problem \cite{abadi2003aurora,golab2012scalable, he2010comet, ghanem2010supporting, KrishnamurthyFDFGLT10}. 
In Spark Streaming, Zaharia et~al. studied how they could exploit in-memory materialization~\cite{zaharia2012discretized}, and in MonetDB, Liarou et~al. studied how ideas from columnar storage can be applied to enable real-time analytics \cite{liarou2012monetdb}.
These works focus on correctness, consistency, and fault tolerance of materialized view maintenance.
\svc proposes an alternative model for view maintenance where we allow approximation error (with guarantees) for queries on materialized views for vastly reduced maintenance time.
In many decision problems, exact results are not needed as long as the probability of error is boundable. 
Sampling has been well studied in the context of query processing~\cite{AgarwalMPMMS13, olken1993random, garofalakis2001approximate}. 
Both the problems of efficiently sampling relations \cite{olken1993random} and processing complex queries \cite{agarwalknowing}, have been well studied. 
In \svc, we look at a new problem, where we efficiently sample from a maintenance strategy, a relational expression that updates a materialized view.
We generalize uniform sampling procedures to work in this new context using lineage \cite{DBLP:journals/vldb/CuiW03} and hashing.
We look the problem of approximate query processing \cite{AgarwalMPMMS13, agarwalknowing} from a different perspective by estimating a ``correction'' rather than estimating query results. 
Srinivasan and Carey studied a problem related to query correction which they called compensation-based query processing \cite{srinivasanC92} for concurrency control but did not study this for sampled estimates.
This work was applied in the context of concurrency control.
However, this work did not consider applications when the correction was applied to a sample as in \svc.
The sampling in \svc introduces new challenges such as sensitivity to outliers, questions of bias, and estimate optimality. 

Sampling has also been studied from the perspective of maintaining samples \cite{DBLP:conf/icde/OlkenR92}.
In~\cite{joshi2008materialized}, Joshi and Jermaine studied indexed materialized views that are amenable to random sampling.
While similar in spirit (queries on the view are approximate), the goal of this work was to optimize query processing and not to address the cost of incremental maintenance.
There has been work using sampled views in a limited context of cardinality estimation \cite{larson2007cardinality}, which is the special case of our framework, namely, the \countfunc query.
Nirkhiwale et al. \cite{DBLP:journals/pvldb/NirkhiwaleDJ13}, studied an algebra for estimating confidence intervals in aggregate queries.
The objective of this work is not sampling efficiency, as in \svc, but estimation.
As a special case, where we consider only views constructed from select and project operators, \svc's hash pushdown will yield the same results as their model.
There has been theoretical work on the maintenance of approximate histograms, synopses, and sketches ~\cite{gibbons1997fast, DBLP:journals/ftdb/CormodeGHJ12}, which closely resemble aggregate materialized views.
The objectives of this work (including techniques such as sketching and approximate counting) have been to reduce the required storage, not to reduce the required update time.

Meliou et al. \cite{DBLP:conf/sigmod/MeliouGNS11} proposed a technique to trace errors in an MV to base data and find responsible erroneous tuples. 
They do not, however, propose a technique to correct the errors as in \svc.
Correcting general errors as in Meliou et al. is a hard constraint satisfaction problem.
However, in \svc, through our formalization of staleness, we have a model of how updates to the base data (modeled as errors) affect MVs, which allows us to both trace errors and clean them.
Wu and Madden \cite{DBLP:journals/pvldb/0002M13} did propose a model to correct ``outliers" in an MV through deletion of records in the base data.
This is a more restricted model of data cleaning than \svc, where the authors only consider changes to existing rows in an MV (no insertion or deletion) and do not handle the same generality of relational expressions (e.g., nested aggregates).
Challamalla et al. \cite{DBLP:conf/sigmod/ChalamallaIOP14} proposed an approximate technique for specifying errors as constraints on a materialized view and proposing changes to the base data such that these constraints can be satisfied.
While complementary, one major difference between the three works \cite{DBLP:conf/sigmod/MeliouGNS11, DBLP:journals/pvldb/0002M13, DBLP:conf/sigmod/ChalamallaIOP14} and \svc is that they require an explicit specification of erroneous rows in a materialized view.
Identifying whether a row is erroneous requires materialization and thus specifying the errors is equivalent to full incremental maintenance. 
We use the formalism of a ``maintenance strategy", the relational expression that updates the view, to allow us to sample rows that are not yet materialized.
However, while not directly applicable for staleness, we see \svc as complementary to these works in the dirty data setting. 
The sampling technique proposed in Section 4 of our paper could be used to approximate the data cleaning techniques in \cite{DBLP:conf/sigmod/MeliouGNS11, DBLP:journals/pvldb/0002M13, DBLP:conf/sigmod/ChalamallaIOP14} and this is an exciting avenue of future work.

Sampling has been explored in the streaming community, and a similar idea of sampling from incoming updates has also been applied in stream processing~\cite{tatbul2003load, Garofalakis, rabkin2014aggregation}.
While some of these works studied problems similar to materialization, for example, the JetStream project (Rabkin et al.) looks at how sampling can help with real-time analysis of aggregates.
None of these works formally studied the class views that can benefit from sampling or formalized queries on these views.
However, there are ideas from Rabkin et al. that could be applied in \svc in future work, for example, their description of coarsening operations in aggregates is very similar to our experiments with the ``roll-up'' queries in aggregate views.
There are a variety of other efforts proposing storage efficient processing of aggregate queries on streams \cite{dobra2002processing, greenwald2001space} which is similar to our problem setting and motivation.


\vspace{-1em}
\section{Limitations and Opportunities}\vspace{-.3em}\label{sec:disc}
While our experiments show that \svc works for a variety of applications, there are a few limitations which we summarize in this section.
There are two primary limitations for \svc: class of queries and types of materialized views.
In this work, we primarily focused on aggregate queries and showed that accuracy decreases as the selectivity of the query increases.
Sampled-based methods are fundamentally limited in the way they can support ``point lookup" queries that select a single row.
This is predicted by our theoretical result that accuracy decreases with $\frac{1}{p}$ where $p$ is the fraction of rows that satisfy the predicate.
In terms of more view definitions, \svc does not support views with ordering or ``top-k'' clauses, as our sampling assumes no ordering on the rows of the MV and it is not clear how sampling commutes with general ordering operations.
In the future, we will explore maintenance optimizations proposed in recent work.
For example, DBToaster has two main components, higher-order delta processing and a SQL query compiler, both of which are complementary to \svc.
\svc proposes a new approach for accurate query processing with MVs.
Our results are promising and suggest many avenues for future work.
In particular, we are interested in deeper exploration of the multiple MV setting.
There are many interesting design problems such as given storage constraints and throughput demands, optimize sampling ratios over all views.
Furthermore, there is an interesting challenge about queries that join mutliple sample MVs managed by \svc.
We are also interested in the possibility of sharing computation between MVs and maintenance on views derived from other views.
Finally, our results suggest relatively a straight forward implementation of adaptive selection of the parameters in \svc such as the view sampling ratio and the outlier index threshold.

\section{Conclusion}\label{conclusion}
\vspace{-.3em}
Materialized view maintenance is often expensive, and in practice, eager view maintenance is often avoided due to its costs.
This leads to stale materialized views which have incorrect, missing, and superfluous rows.
In this work, we formalize the problem of staleness and view maintenance as a data cleaning problem.
\svc uses a sample-based data cleaning approach to get accurate query results that reflect the most recent data for a greatly reduced computational cost.
To achieve this, we significantly extended our prior work in data cleaning, SampleClean \cite{wang1999sample}, for efficient cleaning of stale MVs. 
This included processing a wider set of aggregate queries, handling missing data errors, and proving for which queries optimality of the estimates hold.
We presented both empirical and theoretical results showing that our sample data cleaning approach is significantly less expensive than full view maintenance for a large class of materialized views, while still providing accurate aggregate query answers that reflect the most recent data.

Our results are promising and suggest many avenues for future work.
In this work, we focused on aggregate queries and showed that accuracy decreases as the selectivity of the query increases.
Sampled-based methods are fundamentally limited in the way they can support ``point lookup" queries that select a single row, and we believe we can address this problem with new results in non-parametric machine learning instead of using single-parameter estimators.
In particular, we are interested in deeper exploration of the multiple MV setting.
There are also many interesting design problems such as given storage constraints and throughput demands, optimize sampling ratios over all views.

\textbf{\small We thank Kai Zeng for his advice and feedback on this paper. This research is supported in part by NSF CISE Expeditions Award CCF-1139158, LBNL Award 7076018, and DARPA XData Award FA8750-12-2-0331, and gifts from Amazon Web Services, Google, SAP, The Thomas and Stacey Siebel Foundation, Adatao, Adobe, Apple, Inc., Blue Goji, Bosch, C3Energy, Cisco, Cray, Cloudera, EMC2, Ericsson, Facebook, Guavus, HP, Huawei, Informatica, Intel, Microsoft, NetApp, Pivotal, Samsung, Schlumberger, Splunk, Virdata and VMware.}

\bibliographystyle{abbrv}
\fontsize{8.3pt}{8.5pt} \selectfont
\bibliographystyle{abbrv}
\bibliography{ref} 

\normalsize
\selectfont

\section{Appendix}

\subsection{Extensions}
\subsubsection{MIN and MAX}
\minfunc and \maxfunc fall into their own category since this is a canonical case where bootstrap fails.
We devise an estimation procedure that corrects these queries.
However, we can only achieve bound that has a slightly different interpretation than the confidence intervals seen before.
We can calculate the probability that a larger (or smaller) element exists in the unsampled view.

We devise the following correction estimate for \maxfunc: (1) For all rows in both $S$ and $S'$, calculate the row-by-row difference, (2) let $c$ be the max difference, and (3) add $c$ to the max of the stale view.

We can give weak bounds on the results using Cantelli's Inequality.
If $X$ is a random variable with mean $\mu_x$ and variance $var(X)$, then the probability that $X$ is larger than a constant $\epsilon$ 
\[
\mathbb{P}(X \ge \epsilon + \mu_x ) \le \frac{var(X)}{var(X) + \epsilon^2}
\]
Therefore, if we set $\epsilon$ to be the difference between max value estimate and the average value, we can calculate the probability that we will see a higher value. 

The same estimator can be modified for \minfunc, with a corresponding bound:
\[
\mathbb{P}(X \le \mu_x - a )) \le \frac{var(x)}{var(x) + a^2}
\]
This bound has a slightly different interpretation than the confidence intervals seen before.
This gives the probability that a larger (or smaller) element exists in the unsampled view.

\vspace{-.25em}
\subsubsection{Select Queries}
In \svc, we also explore how to extend this correction procedure to Select queries.
Suppose, we have a Select query with a predicate:
\begin{lstlisting} [mathescape]
SELECT $*$ FROM View WHERE Condition(A);
\end{lstlisting}

We first run the Select query on the stale view, and this returns a set of rows.
This result has three types of data error: rows that are missing, rows that are falsely included, and rows whose values are incorrect.

As in the \sumfunc, \countfunc, and \avgfunc query case, we can apply the query to the sample of the up-to-date view.
From this sample, using our lineage defined earlier, we can quickly identify which rows were added, updated, and deleted.
For the updated rows in the sample, we overwrite the out-of-date rows in the stale query result.
For the new rows, we take a union of the sampled selection and the updated stale selection.
For the missing rows, we remove them from the stale selection.
To quantify the approximation error, we can rewrite the Select query as \countfunc to get an estimate of number of rows that were updated, added, or deleted (thus three ``confidence'' intervals).

\subsection{Extended Proofs}

\subsection{Is Hashing Equivalent To RNG?}
In this work, we argue that hashing can be used for ``sampling" a relational expression.
However, from a complexity theory perspective, hashing is not equivalent to random number generation (RNG).
The existence of true one-way hash functions is a conjecture that would imply $P \ne NP$.
This conjecture is often taken as an assumption in Cryptography.
Of course, the ideal one-way hash functions required by the theory do not exist in practice. However, we find that existing hashes (e.g., linear hashes and SHA1) are sufficiently close to ideal that they can still take advantage of this theory. 
On the other hand, a SHA1 hash is nearly an order of magnitude slower but is much more uniform.
This assumption is called the Simple Uniform Hashing Assumption (SUHA) \cite{cormenintroduction}, and is widely used to analyze the performance of hash tables and hash partitioning.
There is an interesting tradeoff between the latency in computing a hash compared to its uniformity. For example, a linear hash stored procedure in MySQL is nearly as fast pseudorandom number generation that would be used in a TABLESAMPLE operator, however this hash exhibits some non-uniformity. 

\subsubsection{Hashing and Correspondence}
A benefit of deterministic hashing is that when applied in conjunction to the primary keys of a view, we get the Correspondence Property (Definition \ref{correspondence}) for free.
\begin{proposition}[Hashing Correspondence]
Suppose we have $S$ which is the stale view and $S'$ which is the up-to-date view.
Both these views have the same schema and a primary key $a$.
Let $\eta_{a, m}$ be our hash function that applies the hashing to the primary key $a$.
\[
\hat{S} = \eta_{a, m}(S)
\]
\[
\hat{S'} = \eta_{a, m}(S')
\]
Then, two samples $\hat{S'}$ and $\hat{S}$ correspond.
\end{proposition}
\begin{proof}
There are four conditions for correspondence:
\begin{itemize}
\item (1) Uniformity: $\widehat{S'}$ and $\widehat{S}$ are uniform random samples of $S'$ and $S$ respectively with a sampling ratio of $m$
\item (2) Removal of Superfluous Rows: $D = \{\forall s \in \widehat{S} \nexists s' \in S': s(u) = s'(u)\}$, $D \cap \widehat{S'} = \emptyset$ 
\item (3) Sampling of Missing Rows: $I = \{\forall s' \in \widehat{S'} \nexists s \in S: s(u) = s'(u)\}$, $\mathbb{E}(\mid I \cap \widehat{S'} \mid) = m\mid I \mid $ 
\item (4) Key Preservation for Updated Rows: For all $s\in \widehat{S}$ and not in $D$ or $I$, $s' \in \widehat{S}': s'(u) = s(u)$.
\end{itemize}
Uniformity is satisfied under by definition under SUHA (Simple Uniform Hashing Assumption).
Condition 2 is satisfied since if $r$ is deleted, then $r \not \in S'$ which implies that $r \not\in \hat{S'}$.
Condition 3 is just the converse of 2 so it is satisfied.
Condition 4 is satisfied since if $r$ is in $\hat{S}$ then it was sampled, and then since the primary key is consistent between $S$ and $S'$ it will also be sampled in $\hat{S'}$.
\end{proof}

\subsection{Theorem 1 Proof}
\begin{theorem}
Given a derived relation $R$, primary key $a$, and the sample $\eta_{a, m}(R)$.
Let $S$ be the sample created by applying $\eta_{a, m}$ without push down and 
$S'$ be the sample created by applying the push down rules to $\eta_{a, m}(R)$.
$S$ and $S'$ are identical samples with sampling ratio $m$.
\end{theorem}
\begin{proof}
We can prove this by induction.
The base case is where the expression tree is only one node, trivially making this true.
Then, we can induct considering one level of operators in the tree.
$\sigma, \cup, \cap, -$ clearly commute with hashing the key $a$ allowing for push down.
$\Pi$ commutes only if $a$ is in the projection.
For $\bowtie$, a sampling operator on $Q$ can be pushed down if $a$ is in either $k_r$ or $k_s$, or if there is a constraint that links $k_r$ to $k_s$.
There are two cases in which this happens a foreign-key relationship or an equality join on the same key.
For group by aggregates, if $a$ is in the group clause (i.e., it is in the aggregate) then a hash of the operand filters all rows that have $a$ which is sufficient to materialize the derived row.
It is provably NP-Hard to pushdown through a nested group by aggregate such as:
\begin{lstlisting}
SELECT c, count(1)
FROM ( 
       SELECT videoId, sum(1) as c FROM Log 
       GROUP BY videoId
     )
GROUP BY c
\end{lstlisting}
by reduction to a SUBSET-SUM problem.
\end{proof}

\subsection{More about the Hash Operator}
We defined a concept of tuple-lineage with primary keys.
However, a curious property of the deterministic hashing technique is that we can actually hash any attribute while retain the important statistical properties.
This is because a uniformly random sample of any attribute (possibly not unique) still includes every individual row with the same probability.  
A consequence of this is that we can push down the hashing operator through arbitrary equality joins (not just many-to-one) by hashing the join key.

We defer further exploration of this property to future work as it introduces new tradeoffs.
For example, sampling on a non-unique key, while unbiased in expectation, has higher variance in the size of the sample.
Happening to hash a large group may lead to decreased performance. 

Suppose our keys are duplicated $\mu_k$ times on average with variance $\sigma_k^2$, then the variance of the
sample size is for sampling fraction $m$:
\[m(1-m)\mu_k^2+(1-m)\sigma_k^2\]
This equation is derived from the formula for the variance of a mixture distribution.
In this setting, our sampling would have to consider this variance against the benefits of pushing the hash operator further down the query tree. 

\subsection{Experimental Details}

\subsubsection{Join View TPCD Queries}
In our first experiment, we materialize the join of lineitem and orders.
We treat the TPCD queries as queries on the view, and we selected 12 out of the 22 to include in our experiments.
The other 10 queries did not make use of the join.

\subsubsection{Conviva Views}
In this workload, there were annotated summary statistics queries, and we filtered for the most common types.
While, we cannot give the details of the queries, we can present some of the high-level characteristics of 8 summary-statistics type views. 
\begin{itemize} 
\item \textbf{V1.} Counts of various error types grouped by resources, users, date
\item \textbf{V2.} Sum of bytes transferred grouped by resource, users, date
\item \textbf{V3.} Counts of visits grouped by an expression of resource tags, users, date.
\item \textbf{V4.} Nested query that groups users from similar regions/service providers together then aggregates statistics
\item \textbf{V5.} Nested query that groups users from similar regions/service providers together then aggregates error types
\item \textbf{V6.} Union query that is filtered on a subset of resources and aggregates visits and bytes transferred
\item \textbf{V7.} Aggregate network statistics group by resources, users, date with many aggregates.
\item \textbf{V8.} Aggregate visit statistics group by resources, users, date with many aggregates.
\end{itemize}

\subsubsection{Data Cube Specification}
We defined the base cube as a materialized view:
\begin{lstlisting}
select
  sum(l_extendedprice * (1 - l_discount)) as revenue,
  c_custkey, n_nationkey,
  r_regionkey, L_PARTKEY
from
  lineitem, orders,
  customer, nation,
  region
where
  l_orderkey = o_orderkey and
  O_CUSTKEY = c_custkey and
  c_nationkey = n_nationkey and
  N_REGIONKEY = r_regionkey

group by
  c_custkey, n_nationkey, 
  r_regionkey, L_PARTKEY
\end{lstlisting}

Each of queries was an aggregate over subsets of the dimensions of the cube, 
with a \sumfunc over the revenue column.
\begin{itemize}
\item Q1. all
\item Q2. c\_custkey
\item Q3. n\_nationkey
\item Q4. r\_regionkey
\item Q5. l\_partkey
\item Q6. c\_custkey,n\_nationkey
\item Q7. c\_custkey,r\_regionkey
\item Q8. c\_custkey,l\_partkey
\item Q9. n\_nationkey, r\_regionkey
\item Q10. n\_nationkey, l\_partkey
\item Q11. c\_custkey,n\_nationkey, r\_regionkey
\item Q12. c\_custkey,n\_nationkey,l\_partkey
\item Q13. n\_nationkey,r\_regionkey,l\_partkey
\end{itemize}

When we experimented with the median query, we changed the \sumfunc to a median of the revenues.

\subsubsection{Table Of TPCD Queries 2}
We denormalize the TPCD schema and treat each of the 22 queries as views on the denormalized schema.
In our experiments, we evaluate 10 of these with SVC. 
Here, we provide a table of the queries and reasons why a query was not suitable for our experiments.
The main reason a query was not used was because the cardinality of the result was small.
Since we sample from the view, if the result was small eg. < 10, it would not make sense to apply SVC.
Furthermore, in the TPCD specification the only tables that are affected by updates are lineitem and orders; and queries that
do not depend on these tables do not change; thus there is no need for maintenance.

Listed below are excluded queries and reasons for their exclusion.
\begin{itemize}
\item Query 1. Result cardinality too small
\item Query 2. The query was static
\item Query 6. Result cardinality too small
\item Query 7. Result cardinality too small
\item Query 8. Result cardinality too small
\item Query 11. The query was static 
\item Query 12. Result cardinality too small
\item Query 14. Result cardinality too small
\item Query 15. The query contains an inner query, which we treat as a view.
\item Query 16. The query was static 
\item Query 17. Result cardinality too small
\item Query 19. Result cardinality too small
\item Query 20. Result cardinality too small
\end{itemize}

\end{document}